%% file: 0-moncal.tex
\documentclass{article}
 
\input{moncal-macros}

\begin{document}

\title{Monoidal computer III: \\
	A coalgebraic view of \\ computability and complexity\thanks{Partially supported by AFOSR and NSF.}}

\author{Dusko Pavlovic
	\hspace{5em} Muzamil Yahia \\ University of Hawaii, Honolulu, USA\\
  \texttt{\{dusko,muzamil\}@hawaii.edu}}

\date{}
\maketitle

\begin{abstract}
Monoidal computer is a categorical model of \emph{intensional}\/ computation, where many different programs correspond to the same input-output behavior. The upshot of yet another model of computation is that a categorical formalism should provide a high-level language for theory of computation, flexible enough to allow abstracting away the low level implementation details when they are irrelevant, or taking them into account when they are genuinely needed. A salient feature of the approach through monoidal categories is the formal graphical language of \emph{string diagrams}, which supports geometric reasoning about programs and computations.  In the present paper, we provide a coalgebraic characterization of monoidal computer. It turns out that the availability of interpreters and specializers, that make a monoidal category into a monoidal computer, is equivalent with the existence of a \emph{universal state space}, that carries a weakly final state machine for all types of input and output. Being able to program state machines in monoidal computers allows us to represent Turing machines, and capture the time and space needed for their executions. The coalgebraic view of monoidal computer thus provides a convenient  diagrammatic language for studying not only computability, but also complexity.
\end{abstract}

\bigskip
\section{Introduction}
\input{1-moncal-intro.tex}

\bigskip
\section{
Preliminaries}\label{Sec-two}
\input{2-moncal-data.tex}

\bigskip
\section{
Monoidal computer}\label{Sec-three}
\input{3-moncal-moncom.tex}

\bigskip
\section{Coalgebraic view
}\label{Sec-four}
\input{4-moncal-coalg.tex}

\bigskip
\section{Computability
}\label{Sec-five}
\input{5-moncal-turing.tex}

\section{Complexity}\label{Sec-complexity}
\input{6-moncal-complexity.tex}

\bigskip
\section{Final comments}\label{Sec-conclusion}
\input{7-moncal-outro.tex}

\bibliographystyle{plain}
\bibliography{moncal-ref,CT,PavlovicD,ref-moncomp-1,ref-14-LICS,logic,ref-ITCS,semantics,biology}


\appendix
\input{8-moncal-appendix.tex}


\end{document}

%% file: moncal-macros.tex
\usepackage{etex}
\usepackage{tikz-cd}
\usetikzlibrary{decorations.pathmorphing}
\input{tikz-cd-migration}
\tikzcdset{
  shift left/.default=+0.7ex,
  shift right/.default=+0.7ex}
  
  \usepackage{pstricks}
  \usepackage{bbold}
  \usepackage{bbm}

\newenvironment{tikzar}[1]{{}\kern-4pt\begin{tikzcd}[ampersand replacement=\&,#1]}%
{\end{tikzcd}\kern-4pt{}}

\usepackage{graphicx}

\usepackage
{stmaryrd}

\usepackage[USenglish]{babel}
\usepackage{amssymb}
 \usepackage{amstext}
 \usepackage{amsmath}
\usepackage{amsthm}
\usepackage{enumerate}

\input{prooftree}

\renewcommand{\to}{\xrightarrow{}}
\newcommand{\tto}[1]{\xrightarrow{#1}}
\newcommand{\oot}[1]{\xleftarrow{#1}}

\newcommand{\epi}{\twoheadrightarrow}

\newcommand{\mmono}[1]{\stackrel{#1}\rightarrowtail}
\newcommand{\eepi}[1]{\stackrel{#1}\twoheadrightarrow}

\newcommand{\iinclusion}[1]{\stackrel{#1}\hookrightarrow}

\newcommand{\pfn}{\rightharpoonup}
\newcommand{\ppfn}[1]{\stackrel{#1}\rightharpoonup}

\newcommand{\xxp}[2]{\left[{#1},{#2}\right]}
\newcommand{\xxpp}[2]{\xxp{{#1}^+}{#2}}

\newcommand{\tleft}{\triangleleft}
\newcommand{\tright}{\triangleright}
\newcommand{\tstall}{{\scriptstyle \Box}}

\newcommand{\Set}{\mathsf{Set}}
\newcommand{\Pfn}{{\sf Pfn}}

\newcommand{\Rel}{{\sf Rel}}

\newcommand{\id}{{\rm id}}
\newcommand{\blank}{\sqcup}

\newcommand{\cmn}{\Delta}
\newcommand{\cun}{\top}

\newcommand{\cata}[1]{\llparenthesis {#1} \rrparenthesis}
\newcommand{\ana}[1]{\left\llbracket {#1} \right\rrbracket}
\newcommand{\run}[1]{\left\{ {#1} \right\}}
\newcommand{\Run}[1]{\left\{\!\lvert {#1} \rvert\!\right\}}
\newcommand{\RRun}{\left\{\!\lvert \, \rvert\!\right\}}

\newcommand{\UK}{\{ \}}
\newcommand{\GK}{\gamma}

\newcommand{\UKK}[2]{\left\{{#1}\right\}{#2}}

\newcommand{\SK}{[\, ]}
\newcommand{\SKK}[2]{\left[{#1}, {#2}\right]}

\newcommand{\enc}{e}
\newcommand{\dec}{d}

\newcommand{\PP}{\PPp}
\newcommand{\WW}{{\widetilde{\Sigma}}}

\newcommand{\tu}{{\rho}}
\newcommand{\tuu}{\widetilde{\tu}}
\newcommand{\tuuu}{\overline{\tu}}

\newcommand{\tum}{{p}}
\newcommand{\tumm}{\widetilde{\tum}}
\newcommand{\tummm}{\overline \tum}
\newcommand{\Turmm}{\widetilde P}
\newcommand{\Turmmm}{\overline P}
\newcommand{\Turm}{{\QQq}}

\newcommand{\tim}{t}
\newcommand{\timm}{\widetilde \tim}
\newcommand{\timmm}{\overline \tim}
\newcommand{\Timm}{\widetilde T}
\newcommand{\Timmm}{\overline T}

\newcommand{\spac}{s}
\newcommand{\spacc}{\widetilde \spac}
\newcommand{\spaccc}{\overline \spac}
\newcommand{\Spacc}{\widetilde S}
\newcommand{\Spaccc}{\overline S}

\newcommand{\RR}{\SSs}
\newcommand{\CC}{{\sf C}}

\newcommand{\doublone}{\mathbb{1}}

\newcommand{\Nn}{\PP}

\newcommand{\CCC}{{\cal C}}

\newcommand{\OOO}{{\cal O}}

\newcommand{\SSS}{{\cal S}}
\newcommand{\TTT}{{\cal T}}

\renewcommand{\Bbb}{\mathbb}

\newcommand{\CCc}{{\Bbb C}}

\newcommand{\NNn}{{\Bbb N}}

\newcommand{\PPp}{{\Bbb P}}
\newcommand{\QQq}{{\Bbb Q}}

\newcommand{\SSs}{{\Bbb S}}

\newcommand{\ZZz}{{\Bbb Z}}

\newcommand{\OOne}{\mathbbm{1}}

\mathcode`\<="4268 
\mathcode`\>="5269 
\mathchardef\gt="313E 
\mathchardef\lt="313C 

\newcommand{\beq}{\begin{equation}}
\newcommand{\eeq}{\end{equation}}
\newcommand{\ba}[1]{\begin{array}{#1}}
\newcommand{\ea}{\end{array}}
\newcommand{\bea}{\begin{eqnarray}}
\newcommand{\eea}{\end{eqnarray}}
\newcommand{\bear}{\begin{eqnarray*}}
\newcommand{\eear}{\end{eqnarray*}}

\theoremstyle{plain}
\newtheorem{theorem}{Theorem}[section]

\newtheorem*{proposition*}{Proposition}

\newtheorem{proposition}[theorem]{Proposition}
\newtheorem{definition}[theorem]{Definition}
\newtheorem{corollary}[theorem]{Corollary}

\theoremstyle{remark}
\newtheorem*{remark}{Remark}

\theoremstyle{remark}
\newtheorem*{notation}{Notation}
\theoremstyle{remark}

\theoremstyle{remark}

\newcommand{\tot}[1]{{#1}^\bullet}

\newcommand{\comp}{\, ;}

\newcommand{\supp}{{\sf supp}}

\newcommand{\halts}{\!\downarrow}

\newcommand{\TM}{\TTT}

\newcommand{\iif}{\mbox{\it if}}
\newcommand{\true}{\mathtt{t}}
\newcommand{\false}{\mathtt{f}}

\newenvironment{prf}[1]{\begin{trivlist} \item[{\bf ~Proof #1.}]}%
{\qed\end{trivlist}}

\newcommand{\bpr}{\begin{prf}{\!\!}}
\newcommand{\epr}{\end{prf}}
\newcommand{\bprf}[1]{\begin{prf}{#1}}
\newcommand{\eprf}{\end{prf}}

%% file: tikz-cd-migration.tex
\tikzset{
  phantom/.style={draw=none,
    commutative diagrams/every label/.append style={/tikz/auto=false}
  },
}

\def\tikzcdset{\pgfqkeys{/tikz/commutative diagrams}}

\makeatletter

\def\tikzcd@setanchor#1[#2]#3\relax{%
  \ifx\relax#2\relax\else%
    \tikzcdset{@#1transform/.append style={#2},@shiftabletopath}%
  \fi%
  \ifx\relax#3\relax%
    \pgfutil@namelet{tikzcd@#1anchor}{pgfutil@empty}%
  \else%
    \pgfutil@namedef{tikzcd@#1anchor}{.#3}%
  \fi}

\tikzcdset{
  @shiftabletopath/.style={
    /tikz/execute at begin to={%
      \begingroup%
        \def\tikz@tonodes{coordinate[pos=0,commutative diagrams/@starttransform/.try](tikzcd@nodea) %
                          coordinate[pos=1,commutative diagrams/@endtransform/.try](tikzcd@nodeb)}%
        \path (\tikztostart) \tikz@to@path;%
      \endgroup%
      \def\tikztostart{tikzcd@nodea}%
      \tikzset{insert path={(tikzcd@nodea)}}},
    /tikz/commutative diagrams/@shiftabletopath/.code={}},
  start anchor/.code={%
    \pgfutil@ifnextchar[{\tikzcd@setanchor{start}}{\tikzcd@setanchor{start}[]}#1\relax},
  end anchor/.code={%
    \pgfutil@ifnextchar[{\tikzcd@setanchor{end}}{\tikzcd@setanchor{end}[]}#1\relax},
  start anchor=,
  end anchor=,
  shift left/.style={
    /tikz/commutative diagrams/@shiftabletopath,
    /tikz/execute at begin to={%
      \pgfpointnormalised{%
        \pgfpointdiff{\pgfpointanchor{tikzcd@nodeb}{center}}{\pgfpointanchor{tikzcd@nodea}{center}}}%
      \pgfgetlastxy{\tikzcd@x}{\tikzcd@y}%
      \pgfmathparse{#1}%
      \coordinate (tikzcd@nodea) at ([shift={(\pgfmathresult*\tikzcd@y,-\pgfmathresult*\tikzcd@x)}]tikzcd@nodea);%
      \coordinate (tikzcd@nodeb) at ([shift={(\pgfmathresult*\tikzcd@y,-\pgfmathresult*\tikzcd@x)}]tikzcd@nodeb);%
      \tikzset{insert path={(tikzcd@nodea)}}}},
  shift right/.style={
    /tikz/commutative diagrams/shift left={-(#1)}},
  transform nodes/.style={
    /tikz/commutative diagrams/@shiftabletopath,
    /tikz/commutative diagrams/@starttransform/.append style={#1},
    /tikz/commutative diagrams/@endtransform/.append style={#1}},
  shift/.style={
    /tikz/shift={#1},
    /tikz/commutative diagrams/transform nodes={/tikz/shift={#1}}},
  xshift/.style={
    /tikz/xshift={#1},
    /tikz/commutative diagrams/transform nodes={/tikz/xshift={#1}}},
  yshift/.style={
    /tikz/yshift={#1},
    /tikz/commutative diagrams/transform nodes={/tikz/yshift={#1}}}}

\makeatother

%% file: prooftree.tex
\newdimen\proofrulebreadth \proofrulebreadth=.05em
\newdimen\proofdotseparation \proofdotseparation=1.25ex
\newdimen\proofrulebaseline \proofrulebaseline=2ex
\newcount\proofdotnumber \proofdotnumber=3
\let\then\relax
\def\hfi{\hskip0pt plus.0001fil}
\mathchardef\squigto="3A3B
\newif\ifinsideprooftree\insideprooftreefalse
\newif\ifonleftofproofrule\onleftofproofrulefalse
\newif\ifproofdots\proofdotsfalse
\newif\ifdoubleproof\doubleprooffalse
\let\wereinproofbit\relax
\newdimen\shortenproofleft
\newdimen\shortenproofright
\newdimen\proofbelowshift
\newbox\proofabove
\newbox\proofbelow
\newbox\proofrulename
\def\shiftproofbelow{\let\next\relax\afterassignment\setshiftproofbelow\dimen0 }
\def\shiftproofbelowneg{\def\next{\multiply\dimen0 by-1 }%
\afterassignment\setshiftproofbelow\dimen0 }
\def\setshiftproofbelow{\next\proofbelowshift=\dimen0 }
\def\setproofrulebreadth{\proofrulebreadth}

\def\prooftree{
\ifnum  \lastpenalty=1
\then   \unpenalty
\else   \onleftofproofrulefalse
\fi
\ifonleftofproofrule
\else   \ifinsideprooftree
        \then   \hskip.5em plus1fil
        \fi
\fi
\bgroup
\setbox\proofbelow=\hbox{}\setbox\proofrulename=\hbox{}%
\let\justifies\proofover\let\leadsto\proofoverdots\let\Justifies\proofoverdbl
\let\using\proofusing\let\[\prooftree
\ifinsideprooftree\let\]\endprooftree\fi
\proofdotsfalse\doubleprooffalse
\let\thickness\setproofrulebreadth
\let\shiftright\shiftproofbelow \let\shift\shiftproofbelow
\let\shiftleft\shiftproofbelowneg
\let\ifwasinsideprooftree\ifinsideprooftree
\insideprooftreetrue
\setbox\proofabove=\hbox\bgroup$\displaystyle 
\let\wereinproofbit\prooftree
\shortenproofleft=0pt \shortenproofright=0pt \proofbelowshift=0pt
\onleftofproofruletrue\penalty1
}

\def\eproofbit{
\ifx    \wereinproofbit\prooftree
\then   \ifcase \lastpenalty
        \then   \shortenproofright=0pt  
        \or     \unpenalty\hfil         
        \or     \unpenalty\unskip       
        \else   \shortenproofright=0pt  
        \fi
\fi
\global\dimen0=\shortenproofleft
\global\dimen1=\shortenproofright
\global\dimen2=\proofrulebreadth
\global\dimen3=\proofbelowshift
\global\dimen4=\proofdotseparation
\global\count255=\proofdotnumber
$\egroup  
\shortenproofleft=\dimen0
\shortenproofright=\dimen1
\proofrulebreadth=\dimen2
\proofbelowshift=\dimen3
\proofdotseparation=\dimen4
\proofdotnumber=\count255
}

\def\proofover{
\eproofbit 
\setbox\proofbelow=\hbox\bgroup 
\let\wereinproofbit\proofover
$\displaystyle
}%
\def\proofoverdbl{
\eproofbit 
\doubleprooftrue
\setbox\proofbelow=\hbox\bgroup 
\let\wereinproofbit\proofoverdbl
$\displaystyle
}%
\def\proofoverdots{
\eproofbit 
\proofdotstrue
\setbox\proofbelow=\hbox\bgroup 
\let\wereinproofbit\proofoverdots
$\displaystyle
}%
\def\proofusing{
\eproofbit 
\setbox\proofrulename=\hbox\bgroup 
\let\wereinproofbit\proofusing
\kern0.3em$
}

\def\endprooftree{
\eproofbit 
  \dimen5 =0pt
\dimen0=\wd\proofabove \advance\dimen0-\shortenproofleft
\advance\dimen0-\shortenproofright
\dimen1=.5\dimen0 \advance\dimen1-.5\wd\proofbelow
\dimen4=\dimen1
\advance\dimen1\proofbelowshift \advance\dimen4-\proofbelowshift
\ifdim  \dimen1<0pt
\then   \advance\shortenproofleft\dimen1
        \advance\dimen0-\dimen1
        \dimen1=0pt
        \ifdim  \shortenproofleft<0pt
        \then   \setbox\proofabove=\hbox{%
                        \kern-\shortenproofleft\unhbox\proofabove}%
                \shortenproofleft=0pt
        \fi
\fi
\ifdim  \dimen4<0pt
\then   \advance\shortenproofright\dimen4
        \advance\dimen0-\dimen4
        \dimen4=0pt
\fi
\ifdim  \shortenproofright<\wd\proofrulename
\then   \shortenproofright=\wd\proofrulename
\fi
\dimen2=\shortenproofleft \advance\dimen2 by\dimen1
\dimen3=\shortenproofright\advance\dimen3 by\dimen4
\ifproofdots
\then
        \dimen6=\shortenproofleft \advance\dimen6 .5\dimen0
        \setbox1=\vbox to\proofdotseparation{\vss\hbox{$\cdot$}\vss}%
        \setbox0=\hbox{%
                \advance\dimen6-.5\wd1
                \kern\dimen6
                $\vcenter to\proofdotnumber\proofdotseparation
                        {\leaders\box1\vfill}$%
                \unhbox\proofrulename}%
\else   \dimen6=\fontdimen22\the\textfont2 
        \dimen7=\dimen6
        \advance\dimen6by.5\proofrulebreadth
        \advance\dimen7by-.5\proofrulebreadth
        \setbox0=\hbox{%
                \kern\shortenproofleft
                \ifdoubleproof
                \then   \hbox to\dimen0{%
                        $\mathsurround0pt\mathord=\mkern-6mu%
                        \cleaders\hbox{$\mkern-2mu=\mkern-2mu$}\hfill
                        \mkern-6mu\mathord=$}%
                \else   \vrule height\dimen6 depth-\dimen7 width\dimen0
                \fi
                \unhbox\proofrulename}%
        \ht0=\dimen6 \dp0=-\dimen7
\fi
\let\doll\relax
\ifwasinsideprooftree
\then   \let\VBOX\vbox
\else   \ifmmode\else$\let\doll=$\fi
        \let\VBOX\vcenter
\fi
\VBOX   {\baselineskip\proofrulebaseline \lineskip.2ex
        \expandafter\lineskiplimit\ifproofdots0ex\else-0.6ex\fi
        \hbox   spread\dimen5   {\hfi\unhbox\proofabove\hfi}%
        \hbox{\box0}%
        \hbox   {\kern\dimen2 \box\proofbelow}}\doll%
\global\dimen2=\dimen2
\global\dimen3=\dimen3
\egroup 
\ifonleftofproofrule
\then   \shortenproofleft=\dimen2
\fi
\shortenproofright=\dimen3
\onleftofproofrulefalse
\ifinsideprooftree
\then   \hskip.5em plus 1fil \penalty2
\fi
}


%% file: 1-moncal-intro.tex

In theory of computation, an \emph{extensional}\/ model reduces computations to their set theoretic extensions, \emph{computable functions}, whereas an \emph{intensional}\/ model also takes into account the multiple \emph{programs}\/ that describe each computable function \cite[II.3]{BarendregtH:book,Odifreddi:one}. 

In computer science, this semantical gamut got refined on the extensional side by \emph{denotational}\/ models, that take into account not just computable functions but also some computational effects, and  on the intensional side by \emph{operational}\/ models, where the meaning of a program is specified up to an operational equivalence \cite{Montanari:book,
WinskelG:book}. Categorical semantics of computation arose from the realization that \emph{cartesian closed}\/ categories provide a simple and effective framework for studying the extensional models \cite{Lambek-Scott:book}. Both denotational and operational semantics naturally developed as extensions of this categorical framework \cite{MoggiE:monads,Plotkin-Turi:LICS97}. 

The goal of the monoidal computer project is to provide categorical semantics of intensional computation. This turns out to be surprisingly simple technically, but subtle conceptually.  
In this section, we describe the structure of monoidal computer informally, and try to explain it in the context of categorical semantics. In the rest of the paper, we spell out some of its features formally, in particular the \emph{coalgebraic}\/ part. 

\subsection{Categorical computability: context and concept}
The step from a cartesian closed category $\CCC$, as an extensional model of computation, to a monoidal computer $\CCc$, as an intensional model, can be summarized as follows:
\bea\label{eq:ccc-moncom}
\prooftree
\begin{tikzar}{column sep=2em}
\CCC\left(X, \xxp A B\right) \arrow[bend left = 7]{r}{\varepsilon^{AB}_{X}}[swap]{\cong} \& \CCC(X\times A, B) \arrow[bend left = 7]{l}{\lambda^{AB}_{X}}
\end{tikzar}
\justifies
\begin{tikzar}{column sep=2em}
\tot \CCc (X, \PP) \arrow[twoheadrightarrow]{r}{\gamma^{AB}_X} \& \CCc (X\otimes A, B)
\end{tikzar}
\endprooftree
\eea
The first line says that a category $\CCC$ is cartesian closed when it has the (cartesian) products $X\times A$ and a family of bijections, natural in $X$ and indexed over the types $A$ and $B$, between the morphisms $X\times A \to B$ and $X\to \xxp A B$. If a morphism $X\times A \tto f B$ is thought of as an $X$-indexed family of computations with the inputs from $A$ and the outputs in $B$, then the corresponding morphism $X\tto{\lambda^{AB}_X(f)} \xxp A B$ can be thought of as the $X$-indexed family of programs for these computations.  
 This structure is the categorical version of the simply typed extensional lambda calculus: $\lambda_X^{AB}$ corresponds to the operation of \emph{abstraction}, whereas $\varepsilon_X^{AB}$ corresponds to the \emph{application} \cite[Part I]{Lambek-Scott:book}. The equation $\varepsilon_X^{AB} \circ \lambda_X^{AB} = \id$ says that if we \emph{abstract}\/ a computation into a program, and then \emph{apply}\/ that program to some data, then we will get the same result as if we executed the original computation on the data. This is the $\beta$-rule of the lambda calculus, the crux of Alonzo Church's representation of program evaluations as function applications of $\lambda$-abstractions \cite{ChurchA:unsolvable}. The equation $\lambda_X^{AB} \circ \varepsilon_X^{AB} = \id$ says that if we apply a program, and then abstract out of the resulting computation a program, then we will get the same program that we started from. This is the $\eta$-rule of the lambda calculus: the extensionality. Dropping the second equation thus corresponds to modeling the \emph{non-extensional}\/ typed lambda calculus, with \emph{weak}\/ exponent types. While this structure was sometimes interpreted as a model of intensional computation, and  interesting results were obtained \cite{HoofmanR:intensional}, the main result was that every such non-extensional model is \emph{essentially extensional}, in the sense that it contains an extensional model as a computable retract \cite{HayashiS:semifunctors}.
In genuinely intensional models, identifying extensionally equivalent programs is not computable. 

 
The structure of a monoidal computer $\CCc$ is displayed in the second line of \eqref{eq:ccc-moncom}. There are \textbf{three changes} with respect to the cartesian closed structure: 
\begin{enumerate}[a)]
\item the bijections $\varepsilon^{AB}_X$ are relaxed to surjections $\gamma^{AB}_X$; 
\item the exponents $\xxp A B$ are replaced with the  type $\PP$ of \emph{programs}, the same for all types $A$ and $B$; and 
\item the product $\times$ is replaced by a tensor $\otimes$, and $\CCc$ is not a cartesian category, \emph{but}\/ $\tot \CCc$ on the left is its largest cartesian subcategory with $\otimes$ as the product.
\end{enumerate}

%
%
%
%
\textbf{Change (a)} means that we have not only dropped the extensionality equation $\lambda_X^{AB} \circ \varepsilon_X^{AB} = \id$, but eliminated the abstraction operation $\lambda_X^{AB}$ altogether. All that is left of the bijection between the abstractions and the applications, displayed in the first line of \eqref{eq:ccc-moncom}, is a surjection from programs to computations, displayed in the second line of \eqref{eq:ccc-moncom}: for every $X$-indexed family of computations $X\otimes A \tto f B$ there is an $X$-indexed family of programs $X\tto F \PP$ such that $f = \gamma_X^{AB}(F)$.  Could we get away with less? No, because the program evaluation $\gamma_X^{AB}$ has a left inverse $\lambda_X^{AB}$ if and only if the model is essentially extensional (i.e., it contains an extensional retract). We will see in Sec.~\ref{sec:defmoncom} that the program evaluation $\gamma_X^{AB}$ is in fact executed by a universal evaluator $\UK^{AB} \in \CCc(\PP\otimes A, B)$, and thus takes the form $\gamma_X^{AB}(F) = \run F^{AB} = f$. 

\textbf{Change (b)} means that all programs are of the same type $\PP$. The central feature of intensional computation is that any program can be applied to any data, and in particular to itself. The main constructions of computability theory depend on this, as we shall see in Sec.~\ref{Sec-fund}. If computations of type $A\to  B$ were encoded by programs of a type depending on $A$ and $B$, let us write it in the form $\lceil A,B\rceil$, then such programs could not be applied to themselves, but they could only be processed by programs typed in the form $\lceil \lceil A, B\rceil, C\rceil$. That is why all programs must be of the same type $\PP$. We will see in Sec.~\ref{Sec-retracts} that this implies that all types must be retracts of $\PP$. This does not imply that the type structure of a monoidal computer can be completely derived from an applicative structure on $\PP$, as an essentially untyped model of computation \cite[I.15-I.17]{Lambek-Scott:book}. The type structure of monoidal computer, can be derived from internal structure of $\PP$ if and only if the model is essentially extensional (i.e., it contains an extensional retract, like before). But where does the monoidal structure come from?

\textbf{Change (c)} makes monoidal computers into monoidal categories, not cartesian. Just like cartesian categories, monoidal computers have the diagonals and the projections for all types, which are necessary for data copying and deleting, as explained in Sec.~\ref{Sec-two}. Unlike in cartesian categories, though, the diagonals and the projections in monoidal computers are \emph{not natural}. The projections are not natural because intensional computations may not terminate: they are not \emph{total}\/ morphisms. The diagonals are not natural when the computations are not deterministic: they are then not \emph{single-valued}\/ as morphisms. While intensional computations can be deterministic, and the diagonals in a monoidal computer can all be natural, if all projections are natural, i.e. if all computations are total, then the model contains an extensional retract. A monoidal computer is thus a cartesian category if and only if it is essentially extensional. That is why a genuinely intensional monoidal computer must be genuinely monoidal. On the other hand, even a computation that is nowhere defined has a program, and programs are always well-defined values. So while the indexed families of intensional computations cannot all be total functions, the corresponding indexed families of programs must all be total functions. That is why the category $\tot \CCc$ on the left in \eqref{eq:ccc-moncom} is different from $\CCc$: it is the largest subcategory of $\CCc$ for which $\otimes$ is the cartesian product.  

In  \textbf{summary}, dropping or weakening any of changes described in (a-c) leads to the same outcome: an essentially extensional model. For a genuinely intensional model it is thus \emph{necessary}\/ to have (c) a genuinely monoidal structure, (b) untyped programs, and (a) no computable program abstraction operators. It was shown in \cite{PavlovicD:IC12,PavlovicD:MonCom} that this is also \emph{sufficient}\/ for a categorical reconstruction of the basic concepts of computability. Sections \ref{Sec-two} and \ref{Sec-three} provide a brief overview of this. But our main concern in this paper is \emph{complexity}.

\subsection{Categorical complexity: Coalgebraic view}
To capture complexity, we must capture dynamics, i.e. access the actual process of computation. This, of course, varies from model to model, and different models of computation induce different notions of complexity. Abstract complexity \cite{BlumM:axioms} provides, in a sense, a model-independent common denominator, which can be viewed as an abstract notion of complexity; but the categorical view of computations as morphisms at the first sight does not even provide a foothold for abstract complexity. We attempted to mitigate the problem by extending the structure of monoidal computer by \emph{grading} \cite{PavlovicD:MonCom2}, but the approach turned out to be impractical for our goals (indicated in the next section). Now it turns out to also be unnecessary, since dynamics of computation can be captured using the \emph{coalgebraic}\/ tools available in any monoidal computer.

Coalgebra is the categorical toolkit for studying dynamics in general \cite{PavlovicD:LICS98,
RuttenJ:universal}, and dynamics of computation in particular \cite{KlinB:TCS,
PavlovicD:AMAST06,Plotkin-Turi:LICS97}. Coalgebras, as morphisms in the form $X\to EX$ for an endofuctor $E$,  provide a categorical view of automata, state machines, and processes with state update \cite{JacobsB:book-coalg,PavlovicD:MSCS15}; the other way around, all coalgebras can be thought of as processes with state update. In the framework on this paper, only a very special class of coalgebras will be considered, as the morphisms in the form $X\times A \to X\times B$, corresponding to what is usually called \emph{Mealy machines} \cite[\ldots]{bonsague08,hvid06,holcombe}. In the presence of the exponents, such morphisms can be transposed to proper coalgebras in the form $X\to \xxp A {X\times B}$. But coalgebra provides a categorical reconstruction of state machines even without the exponents, since the homomorphisms remain the same, and the category of machines is isomorphic to a category of coalgebras even if the objects are not presented as coalgebras in the strict sense. Our "coalgebras" will thus be in the form $X\times A \to X\times B$, or more generally $X\otimes A \to X\otimes B$.

The crucial step in moving the monoidal computer story into the realm of coalgebra is to replace the $X$-indexed functions $X\times A \tto f B$ with $X$-state machines $X\times A \tto m X\times B$. While a function $f$ mapped for each index $x$ an input $a$ to an output $b$, a machine $m$ now maps at each state $x$ an input $a$ to an output $b$, \emph{and updates the state to $x'$}. This state update provides an abstract view of dynamics. Continuous dynamics can be captured in varying the same approach \cite{PavlovicD:MSCS15,PavlovicD:LICS98}. This step from $X$-indexed functions to $X$-state machines is displayed in the first row of Table~\ref{table}.  
\begin{table}[htp]
\begin{center}
\begin{tabular}{|c||c|c|}
\hline
models & static & dynamic \\
\hline 
\hline
\begin{minipage}[b]{1.5cm}
\centering
extensional models:\ \\[1ex]
cartesian closed\\[1.3em]
\hspace{3em}
\end{minipage} & 
\begin{minipage}[b]{4.3cm}
\hspace{1em} $\begin{tikzar}{row sep=2em,column sep=-1em}
\xxp A B \times A \ar{rr}{\varepsilon}   \& \&\ \    B\ \ \ \    \\ 
 \& X \times A \ar{ur}[swap]{\forall f} \ar[dashed]{ul}{\exists !\,  \lambda f \times A} 
 \end{tikzar}
$\\[1ex]
\small abstractions $\underset{\lambda}{\stackrel{\varepsilon}\longleftrightarrow}$
 applications\\[-2ex]
\hspace{2em}
\end{minipage}
& 
\begin{minipage}[b]{4.3cm}
\centering
$\begin{tikzar}{row sep=1.8em,column sep=-2.3em}
\& 
{[{A^{+}},B]}\times B\\   
{[{A^{+}},B]}\times A \ar{ur}{\xi}   \& \& \ \ \  X\times B\ \ \   \ar[dashed]{ul}[swap]{
\ana m \times B} 
 \\ 
 \& \ \ \   X \times A\ \ \   \ar{ur}[swap]{\forall m} \ar[dashed]{ul}{
\exists !\, \ana m \times A} \end{tikzar}
$\\[1ex]
\hspace{.5em} \small behaviors $\oot{\ana -}$ machines\\[-2ex]
\hspace{2em}
\end{minipage}
\\
\hline
\begin{minipage}[b]{1.8cm}
\centering
intensional models:\ \\[1ex]
monoidal computers\ \\[2em]
\hspace{3em}
\end{minipage} & 
\begin{minipage}[b]{4cm}
\centering
$\begin{tikzar}{row sep=2em,column sep=-.5em}
\PP \otimes A \ar{rr}{\UK}   \& \&   B\ \     \\ 
 \& X \otimes A \ar[name=U]{ur}[swap]{\forall f} \ar[dashed,name=G,below]{ul}{\exists F \times A} 
 \end{tikzar}
$
\\[1ex]
\small programs $\to\!\!\!\!\to$ computations\\[-2ex]
\hspace{2em}
\end{minipage}
& 
\begin{minipage}[b]{4.7cm}
\centering
$\begin{tikzar}{row sep=2em,column sep=0em}
\& \PP \otimes B\\   
\PP\otimes A \ar{ur}{\RRun}   \& \&  X\otimes B \ar[dashed]{ul}[swap]{\exists M \otimes B} 
 \\ 
 \& X \otimes A \ar{ur}[swap]{\forall m} \ar[dashed]{ul}{\exists M \otimes A} \end{tikzar}
$\\[1ex]
\small adaptive programs $\pfn \!\!\!\!\pfn$ processes\\[-2ex]
\hspace{2em}
\end{minipage}
\\
\hline
\end{tabular}
\end{center}
\caption{Semantic directions}
\label{table}
\end{table}%
The representation of functions from $A$ to $B$ by the elements of $\xxp A B$ lifts to the representation of machines with inputs in $A$ and outputs in $B$ by the induced behaviors in $\xxpp A B$, where $A^+$ is the inductive type of the nonempty sequences from $A$. Behaviors are thus construed as \emph{functions extended in time} \cite{
JacobsB:book-coalg,
PavlovicD:MonCom,RuttenJ:universal}. 
In the presence of list constructors, the representation of functions using the exponents $\xxp A B$ induces the representation of machines using the final machines $\xxpp A B$. The other way around, the final machines induce the exponents as soon as the idempotents split. This is proved in the Appendix.

The rows of Table~\ref{table} depict the step from static models to dynamic models. The columns depict the step from the extensional to the intensional.  The left-hand column is just a different depiction of  \eqref{eq:ccc-moncom}: the upper triangle unpacks the bijection in the first line of \eqref{eq:ccc-moncom}, whereas the lower triangle unpacks the surjection in the second line. The right-hand column is the step from the extensional coinduction of final state machines to the intensional coinduction in monoidal computer. The bottom row of Table~\ref{table} is the step from the monoidal computer structure presented in terms of universal evaluators, the content of Sec.~\ref{Sec-three}, to the monoidal computer structure presented in terms of \emph{universal processes}, the content of Sec.~\ref{Sec-four}. The fact that the two presentations are equivalent is stated in Thm.~\ref{thm:moncom-coalg}. This coalgebraic view of intensional computation opens an alley towards capturing dynamics of Turing machines in Sec.~\ref{Sec-five}, and a direct internalization of time and space complexity measures in Sec.~\ref{Sec-complexity}. A comment about the role of coalgebra in this effort is in Sec.~\ref{Sec-conclusion}. A general approach through abstract complexity is provided in the full version of the paper.  

\subsection{Background and related work}
While computability and complexity theorists seldom felt a need to learn about categories, there is a rich tradition of categorical research in computability theory, starting from one of the founders of category theory and his students \cite{eilenberg-elgot,heller-dipaola}, through extensive categorical investigations of realizability \cite{HofstraW13,HylandJME:efft,OostenJ:book
}, to the recent work on Turing categories \cite{Cockett-Hofstra}, and on a monoidal structure of Turing machines \cite{BarthaM:monoidal-TM}. A categorical account of time complexity was proposed in \cite{CockettR:complex}, using a special structure called \emph{timed sets}, introduced for the purpose. While our approach in \cite{PavlovicD:MonCom2} used grading in a similar way, our current approach seems closer in spirit to \cite{AspertiA:rice}, even if that work is neither coalgebraic nor explicitly categorical. Our effort originated from a need for a framework for reasoning about logical depth of cryptographic protocols and algorithms \cite{PavlovicD:NSPW11
}. The scope of the project vastly exceeded the original cost estimates \cite{PavlovicD:Lambek14}, but also the original benefit expectations. The unexpectedly simple diagrammatic formalism of monoidal computer has even been used as a slick teaching tool in several courses.\footnote{The course materials are available from \texttt{http://www.asecolab.org/courses/222/}, and the textbook \cite{PavlovicD:MonCom} is in preparation.}

%% file: 2-moncal-data.tex
A monoidal computer is a \emph{symmetric monoidal category}\/ with some additional structure. As a matter of convenience, and with no loss of generality, we assume that it is a \emph{strict}\/ monoidal category. The reader familiar with these concepts may wish to skip to the next section. For the casual reader unfamiliar with these concepts, we attempt to provide enough intuitions to understand the presented ideas. The reader interested to learn more about monoidal categories should consult one of many textbooks, e.g. \cite[VII.1,XI]{MacLaneS:CWM}. 

\subsection{Monoidal categories}
Intuitively, a monoidal category is a category $\CCc$ together with a functorial monoid structure 
$\CCc\times \CCc \tto \otimes \CCc \oot I \doublone$. When $\CCc$ is a monoidal \emph{computer}, then we think of its objects $A, B, \ldots \in |\CCc|$ as datatypes, and of its morphisms, $f, g,\ldots\in \CCc(A,B)$ as computations. The tensor product $A\otimes P\tto{f \otimes t} B\otimes Q$ then captures the parallel composition of the computations $A\tto f B$ and $P\tto t Q$, whereas the categorical composition $A\tto{ f \comp g} C$  is the sequential composition of $A\tto f B$ and $B\tto g C$. 

With no loss of generality, we assume that tensors are \emph{strictly}\/ associative and unitary, and thus treat the objects $A\otimes (B\otimes C)$ and $(A\otimes B)\otimes C$ as the same, and do not distinguish $A\otimes I$ and $I\otimes A$ from $A$. This allows us to elide many parentheses and natural coherences \cite[Sec.~VII.2]{Joyal-Street:geometry,MacLaneS:CWM}. Note, however, that the isomorphisms $A\otimes B \tto \varsigma B\otimes A$ cannot be eliminated without causing a degeneracy. 

\begin{notation} When no confusion seems likely, we write 
\begin{itemize}
\item $AB$ instead of $A\otimes B$
\item $\CCc(X)$ instead of $\CCc(I,X)$
\end{itemize}
We omit the typing superscripts whenever the types are clear from the context. 
\end{notation}

\medskip
\noindent{\bf String diagrams.}  
A salient feature of monoidal categories is that the algebraic laws of the monoidal structure correspond precisely and conveniently to the geometric laws of \emph{string diagrams}, formalized in \cite{Joyal-Street:geometry}, but going back to \cite{Penrose}. See also \cite{selinger2011survey} for a survey. 
A string diagram usually consists of polygons or ovals linked by strings. In a monoidal computer, the polygons represent computations, whereas the strings represent data types, or the channels through which the data of the corresponding types flow. String diagrams thus display the data flows through composite computations. The reason why string diagrams are convenient for this is that the two program operations that usually generate data flows, the sequential composition $f \comp g$ and the parallel composition $f\otimes t$, precisely correspond to the two geometric operations that generate string diagrams: one is the operation of connecting the polygons $A\tto f B$ and $B\tto g C$ by the string $B$, whereas the other one puts the polygons $A\tto f B$ and $P\tto t Q$ next to each other without connecting them. 
The associativity of these geometric operations then imposes the associativity law on the corresponding operations on computations. The identity morphism $\id_A$, as the unit of the sequential composition, can be viewed as the channel of type $A$, and can thus be presented as the string $A$ itself, or as an "invisible polygon" freely moved along the string $A$. The unit type $I$ can be similarly presented as an "invisible string", freely added and removed to string diagrams. The algebraic laws of the monoidal structure are thus captured by the geometric properties of the string diagrams. The string crossings correspond to the symmetries $A\otimes B \tto \varsigma  B\otimes A$. 

\subsection{Data services} 
We call \emph{data service}\/ the monoidal structure that allows passing the data around in a monoidal category. In computer programs and in mathematical formulas, the data are usually passed around using variables. They allow copying and propagating the data values where they are needed, or deleting them when they are not needed. The basic features of a variable are thus that it can be freely copied or deleted. The basic data services over a type $A$ in a monoidal category $\CCc$ are 
\begin{itemize}
\item the \emph{copying}\/ operation $A\tto \cmn A\otimes A$, and 
\item the \emph{deleting}\/ operation $A\tto \cun I$, 
\end{itemize}
which together form a \emph{commutative comonoid}, i.e. 
satisfy the equations

\smallskip
\bigskip
\newcommand{\eqone}{$\mbox{\footnotesize ${\cmn}\comp{(\cmn \otimes A)}  =  {\cmn}\comp{(A\otimes \cmn)}$}$}
\newcommand{\eqtwo}{$\mbox{\footnotesize 
$\cmn\comp (\cun \otimes A) =  \cmn \comp (A\otimes \cun)=  \id_A$}$}
\newcommand{\eqthree}{$\mbox{\footnotesize $\cmn\comp\, \sigma\  \ =\  \ {\cmn}$}$}
\begin{center}
\def\JPicScale{.65}
\input{PIC/comonoid}
\end{center}
The correspondence between variables and comonoids was formalized and explained in \cite{PavlovicD:MSCS97
}. 
The algebraic properties of the binary copying induce unique $n$-ary copying $A\tto\cmn A^{\otimes n}$, for all $n\geq 0$. The tensor products $\otimes$ in $\CCc$ are the cartesian products $\times$ if and only if every $A$ in $\CCc$ carries a canonical comonoid $A\times A\oot \cmn A \tto \cun \doublone$, where $\doublone$ is the final object of $\CCc$, and all morphisms of $\CCc$ are comonoid homomorphisms, or equivalently, the families $A\tto \cmn A\times A$ and $A\tto \cun \doublone$ are natural. Cartesian categories are thus just monoidal categories with natural families of copying and deleting operations.

\begin{definition}\label{def:dataser}
A \emph{data service}\/ of type $A$ in a monoidal category $\CCc$ is a commutative comonoid structure $A\otimes A\oot \cmn A \tto \cun I$, where $\cmn$ provides the \emph{copying}\/ service, and $\cun$ provides the \emph{deleting}\/ service. 
\end{definition}

\paragraph{Examples and non-examples of data services.} 
The cartesian structure of the category $\Set$ of sets and functions provides the standard data services in the monoidal category $\Rel$ of sets and relations and in the monoidal category $\Pfn$ of sets and partial functions. The cartesian products from $\Set$ induce not only the monoidal structure of $\Rel$ and $\Pfn$, but cartesian comonoids $A\times A \oot \cmn A\tto \cun \doublone$ in $\Set$ also induce data services in $\Rel$ and $\Pfn$. The fact that both families of morphisms $\cmn$ and $\cun$, indexed over $A$, are natural with respect to all functions is what makes $\Set$ into a cartesian category. But $\Rel$ is not cartesian because neither $\cmn$ nor $\cun$ are natural with respect to relations; and $\Pfn$ is not natural because $\cun$ is not natural with respect to partial functions (although $\cmn$ is). In addition, $\Rel$ also admits many nonstandard data services, that do not come from the cartesian structure. This was analyzed in \cite{PavlovicD:QI09,PavlovicD:Qabs12}. Any abelian group (or groupoid) structure $A\times A \tto + A$ can be used as the comparison operation, with the corresponding copying operation $A \tto \cmn A\times A$ relating each $x\in A$ with all pairs $<y,z>\in A\times A$ such that $x=y+z$. The deletion operation $A\tto \cun \doublone$ relates the unit of the group $A$ with the only element of $\doublone$. In any case, restricting to the comonoid homomorphisms makes the chosen data services natural, and the resulting subcategory cartesian.

\begin{definition}\label{def:maps}
A morphism $f\in \CCc(A,B)$ is a \emph{map}\/ if it is a comonoid homomorphism with respect to the data services on $A$ and $B$, which means that it satisfies the following equations

\medskip
\begin{center}
\renewcommand{\eqone}{f \comp {\cmn_B}  =\   {\cmn_A}\comp{(f\otimes f)}}
\renewcommand{\eqtwo}{f \comp {\cun_B}\ =\ \cun_A}
\newcommand{\monoidd}{\cmn}
\newcommand{\fun}{\scriptstyle f}
\newcommand{\One}{A}
\newcommand{\Two}{B}
\newcommand{\delete}{\cun}
\def\JPicScale{.9}
\input{PIC/homomorphism}
\end{center}
Given a symmetric monoidal category $\CCc$ with data services, we denote by $\tot \CCc$ the subcategory spanned by the maps with respect to its data services, i.e. by those $\CCc$-morphisms that preserve copying and deleting.
\end{definition}

\begin{remark}If $\CCc$ is the category of relations, then the first equation says that $f$ is a single-valued relation, whereas the second equation says that it is total. Hence the name. Note that the morphisms $\cmn$ and $\cun$ from the data services are maps with respect to the data service that they induce. They are thus contained in $\tot\CCc$, and each of them forms a natural transformation with respect to the maps. This just means that the tensor $\otimes$, restricted to $\tot\CCc$, is the cartesian product.
\end{remark}

%% file: PIC/comonoid.tex
\ifx\JPicScale\undefined\def\JPicScale{1}\fi
\psset{unit=\JPicScale mm}
\psset{linewidth=0.3,dotsep=1,hatchwidth=0.3,hatchsep=1.5,shadowsize=1,dimen=middle}
\psset{dotsize=0.7 2.5,dotscale=1 1,fillcolor=black}
\psset{arrowsize=1 2,arrowlength=1,arrowinset=0.25,tbarsize=0.7 5,bracketlength=0.15,rbracketlength=0.15}
\begin{pspicture}(0,0)(181.88,25)
\rput(122.5,10){$=$}
\rput(27.5,10){$=$}
\rput{0}(43.12,10.62){\psellipse[fillstyle=solid](0,0)(1.56,-1.56)}
\psline(48.12,18.12)
(48.12,17.5)
(48.13,15.62)(44.38,11.87)
\psline(38.12,18.12)
(38.12,17.5)
(38.13,15.62)(41.88,11.87)
\rput{0}(36.88,4.38){\psellipse[fillstyle=solid](0,0)(1.57,-1.57)}
\psline(13.12,9.38)
(18.12,4.38)
(16.88,4.38)(17.5,4.38)
\psline(31.87,17.5)
(31.88,11.25)
(31.88,9.37)(35.63,5.62)
\rput{0}(11.88,10.62){\psellipse[fillstyle=solid](0,0)(1.56,-1.57)}
\rput{0}(17.82,4.69){\psellipse[fillstyle=solid](0,0)(1.56,-1.57)}
\psline(22.5,17.5)
(22.5,15)
(22.5,9.37)(18.75,5.62)
\pscustom[]{\psline(41.88,9.38)(36.88,4.38)
\psbezier(36.88,4.38)(36.88,4.38)(36.88,4.38)
\psbezier(36.88,4.38)(36.88,4.38)(36.88,4.38)
}
\psline(16.87,17.5)
(16.87,16.87)
(16.88,15)(13.13,11.25)
\psline(6.87,17.5)
(6.87,16.87)
(6.88,15)(10.63,11.25)
\rput(92.5,10){$=$}
\rput{0}(111.88,11.25){\psellipse[fillstyle=solid](0,0)(1.57,-1.56)}
\rput{0}(105.62,5){\psellipse[fillstyle=solid](0,0)(1.56,-1.56)}
\psline(74.37,10)
(79.37,5)
(78.13,5)(78.75,5)
\psline(105.63,3.75)
(105.63,0)
(105.63,-0.62)(105.63,0)
\psline(100.62,18.12)
(100.63,11.87)
(100.63,10)(104.38,6.25)
\rput{0}(73.12,11.25){\psellipse[fillstyle=solid](0,0)(1.57,-1.56)}
\rput{0}(79.06,5.31){\psellipse[fillstyle=solid](0,0)(1.56,-1.56)}
\psline(83.75,18.12)
(83.75,15.62)
(83.75,10)(80,6.25)
\psline(78.75,3.75)
(78.75,0)
(78.75,-0.62)(78.75,0)
\pscustom[]{\psline(110.63,10)(105.63,5)
\psbezier(105.63,5)(105.63,5)(105.63,5)
\psbezier(105.63,5)(105.63,5)(105.63,5)
}
\psline(36.88,3.12)
(36.88,0)
(36.88,-0.62)(36.88,0)
\psline(17.5,3.12)
(17.5,0)
(17.5,-0.62)(17.5,0)
\psline(130,18.12)
(130,0)
(130,-0.62)(130,0)
\rput(27.5,25){$\eqone$}
\rput(100,25){$\eqtwo$}
\rput(163.75,10){$=$}
\rput{0}(176.88,5){\psellipse[fillstyle=solid](0,0)(1.56,-1.56)}
\psline(176.88,3.75)
(176.88,0)
(176.88,-0.62)(176.88,0)
\psline(171.88,18.75)
(171.88,11.87)
(171.88,10)(175.63,6.25)
\psline(150,3.75)
(150,0)
(150,-0.62)(150,0)
\psline(145.62,18.12)
(154.38,11.25)
(154.38,10)(150,5.62)
\psline(181.88,18.75)
(181.88,16.25)
(181.88,10.01)(176.88,5.01)
\psline[border=0.7](154.38,18.12)
(145.62,11.25)
(145.62,10)(150,5.62)
\psline(154.38,19.38)(154.38,18.12)
\psline(145.62,19.38)(145.62,18.12)
\rput(159.38,25){$\eqthree$}
\rput{0}(150,5.62){\psellipse[fillstyle=solid](0,0)(1.56,-1.56)}
\end{pspicture}

%% file: PIC/homomorphism.tex
\ifx\JPicScale\undefined\def\JPicScale{1}\fi
\psset{unit=\JPicScale mm}
\psset{linewidth=0.3,dotsep=1,hatchwidth=0.3,hatchsep=1.5,shadowsize=1,dimen=middle}
\psset{dotsize=0.7 2.5,dotscale=1 1,fillcolor=black}
\psset{arrowsize=1 2,arrowlength=1,arrowinset=0.25,tbarsize=0.7 5,bracketlength=0.15,rbracketlength=0.15}
\begin{pspicture}(0,0)(69.06,25)
\rput(11.88,10){$=$}
\pspolygon[](1.88,7.5)(5.62,7.5)(5.62,3.75)(1.88,3.75)
\pspolygon[](16.88,15.01)(20.63,15.01)(20.63,11.25)(16.88,11.25)
\pspolygon[](25.62,15.01)(29.39,15.01)(29.39,11.25)(25.62,11.25)
\rput(3.75,5.62){$\fun$}
\rput(18.75,13.12){$\fun$}
\rput(27.5,13.12){$\fun$}
\rput(61.88,10){$=$}
\pspolygon[](51.88,9.38)(55.62,9.38)(55.62,5.62)(51.88,5.62)
\rput(53.76,7.5){$\fun$}
\rput{0}(3.75,12.5){\psellipse[fillstyle=solid](0,0)(1.56,-1.56)}
\pscustom[]{\psline(3.75,8.12)(3.75,7.5)
\psbezier(3.75,7.5)(3.75,7.5)(3.75,7.5)
\psline(3.75,7.5)(3.75,11.25)
}
\pscustom[]{\psline(3.75,0)(3.75,-0.62)
\psbezier(3.75,-0.62)(3.75,-0.62)(3.75,-0.62)
\psline(3.75,-0.62)(3.75,3.75)
}
\psline(3.75,11.88)
(0.62,15)
(-0.62,16.25)(-0.62,19.38)
\psline(4.38,12.5)
(8.12,16.25)
(8.12,17.5)(8.12,19.38)
\rput{0}(23.12,4.38){\psellipse[fillstyle=solid](0,0)(1.57,-1.57)}
\pscustom[]{\psline(23.12,0)(23.12,-0.62)
\psbezier(23.12,-0.62)(23.12,-0.62)(23.12,-0.62)
\psline(23.12,-0.62)(23.12,3.12)
}
\psline(23.12,3.75)
(20,6.88)
(18.75,8.12)(18.75,11.25)
\psline(23.75,4.38)
(27.5,8.12)
(27.5,9.38)(27.5,11.25)
\pscustom[]{\psline(18.75,15)(18.75,15.62)
\psbezier(18.75,15.62)(18.75,15.62)(18.75,15.62)
\psline(18.75,15.62)(18.75,19.38)
}
\pscustom[]{\psbezier(27.5,15)(27.5,15)(27.5,15)(27.5,15)
\psbezier(27.5,15)(27.5,15)(27.5,15)
\psline(27.5,15)(27.5,19.38)
}
\pscustom[]{\psbezier(53.75,0)(53.75,0)(53.75,0)(53.75,0)
\psbezier(53.75,0)(53.75,0)(53.75,0)
\psline(53.75,0)(53.75,5.62)
}
\rput{90}(53.75,14.38){\psellipse[fillstyle=solid](0,0)(1.57,-1.56)}
\pscustom[]{\psline(53.75,10)(53.75,9.38)
\psbezier(53.75,9.38)(53.75,9.38)(53.75,9.38)
\psline(53.75,9.38)(53.75,13.12)
}
\rput{90}(67.5,14.38){\psellipse[fillstyle=solid](0,0)(1.57,-1.56)}
\psline(67.5,10)
(67.5,0)
(67.5,9.38)(67.5,13.12)
\rput(16.25,25){$\eqone$}
\rput(58.75,25){$\eqtwo$}
\end{pspicture}

%% file: 3-moncal-moncom.tex
\subsection{Evaluation and evaluators}\label{sec:defmoncom}

\begin{definition}\label{def:moncom}
A \emph{monoidal computer}\/ is 
\begin{itemize}
\item a (strict) symmetric monoidal category  $\CCc$, with 
\item a data service $A\otimes A\oot \cmn A \tto \cun I$ on every $A$,
\item  a distinguished \emph{type of programs}\/ $\PP$, 
\item for every pair of types $A, B$ an $X$-natural family of surjections 
$$\begin{tikzar}{row sep=-.25em,column sep=1em}
\tot \CCc(X, \PP) \arrow[twoheadrightarrow]{rrr}{\GK^{AB}_X} \&\&\& \CCc(XA, B)
\end{tikzar}$$
called \emph{program evaluation}. 
\end{itemize}
%
%
\end{definition}

The next proposition says that the program evaluations from Def.~\ref{def:moncom} are a categorical view of Turing's \emph{universal computer} \cite{TuringA:Entscheidung}, or of Kleene's \emph{acceptable enumerations}\/ \cite[II.5]{Hartley:book,Odifreddi:one}, or of \emph{interpreters}\/ and \emph{specializers}\/ from programming language theory \cite{JonesN:book-computability}. 

\begin{proposition}\label{prop:us}
	Let $\CCc$ be a symmetric monoidal category with data services. Then specifying the program evaluations ${\GK^{AB}_X}:\tot \CCc (X, \PP) \epi  \CCc(X A, B)$ that make $\CCc$ into a monoidal computer, as defined in \ref{def:moncom}, is equivalent to giving for any three types $A, B, C\in |\CCc|$ the following two morphisms:

\begin{enumerate}[(a)]

\item \label{itemuniv} a \emph{universal evaluator} 
	$\UK^{AB} \in \CCc(\PP  A,B)$ such that for every computation $f\in \CCc(A, B)$ there is a program $F\in \tot \CCc (\PP)$ with
\medskip
\bear
f(a)\  & = &\ \  \UKK F a\\[3ex] 
 \def\JPicScale{.5}\newcommand{\aah}{\scriptstyle A}\renewcommand{\dh}{\scriptstyle B}\newcommand{\ahh}{f} \input{PIC/uni-eval-f.tex} & = & \def\JPicScale{.5}\newcommand{\aah}{\scriptstyle F}\renewcommand{\dh}{\scriptstyle B}\newcommand{\ahh}{\UK
 }\newcommand{\bhh}{\scriptstyle A} \input{PIC/uni-eval-p.tex}
\eear
\vspace{1\baselineskip}

\item \label{itempart} a
\emph{partial evaluator} 
$\SK^{(AB)C} \in \tot \CCc(\PP  A,\PP)$ such that 
\medskip
\bear
\UKK G {(a,b)}\ \ & = & \ \ \ \UKK{\SKK G a} b\\[3ex]
\def\JPicScale{.55}\renewcommand{\dh}{\scriptstyle \PP}\newcommand{\ahh}{\UK
}\newcommand{\bhh}{\scriptstyle A}\newcommand{\chh}{\scriptstyle B}\newcommand{\Lhh}{\scriptstyle C} \input{PIC/uni-eval-uMN.tex} & = & \def\JPicScale{.55}\renewcommand{\dh}{\scriptstyle \Nn}
\newcommand{\ahh}{\UK
}\newcommand{\shh}{\SK
}\newcommand{\bhh}{\scriptstyle A}\newcommand{\chh}{\scriptstyle B}\newcommand{\Lhh}{\scriptstyle C} \input{PIC/uni-eval-s.tex}
\eear
\vspace{2\baselineskip}
\end{enumerate}
\end{proposition}

\begin{remark} Note that the partial evaluators $\SK$ are maps, i.e. total and single valued morphisms in $\tot \CCc$, whereas the universal evaluators $\UK$ are ordinary morphisms in $\CCc$. A recursion theorist will recognize the universal evaluators as Turing's \emph{universal machines}\/ \cite{TuringA:Entscheidung}, and the partial evaluators as G\"odel's primitive recursive \emph{substitution function} $S$, enshrined in Kleene's  $S^m_{n}$-theorem \cite{KleeneSC:ordinal}. 
A programmer can think of the universal evaluators as \emph{interpreters}, and of the partial evaluators as \emph{specializers} \cite{JonesN:book-computability}\footnote{When the theory is refined, it becomes useful to recognize subtle but important conceptual and technical distinctions.}. 
In any case, (a) can be understood as saying that every computation can be programmed; and then (b) says that any program with several inputs can be evaluated on any of its inputs, and reduced to a program that waits for the remaining inputs:

\[
\newcommand{\computations}{$g(a,b)$}
\newcommand{\programmed}{$\UKK G {(a,b)}$}
\newcommand{\evaluated}{$\UKK{\SKK G a} b$}
\newcommand{\program}{$G$}
\newcommand{\gee}{$g$}
\newcommand{\Progtype}{}
\newcommand{\Progk}{\scriptstyle A}
\newcommand{\Progm}{\scriptstyle B}
\newcommand{\Progn}{\scriptstyle C}
\newcommand{\EQLS}{=}
\newcommand{\universal}{\{\}}
\newcommand{\universalmkn}{\{\}}
\renewcommand{\partial}{[\,]}
\def\JPicScale{.3}
\input{PIC/computation.tex}
\]
\end{remark}

\begin{proof}[Proof of Prop.~\ref{prop:us}] Every natural transformation $\GK^{AB}:\tot \CCc(-,\PP) \to \CCc(- \otimes A, B)$ is uniquely determined by the computation $\UK^{AB} = \GK^{AB}_\PP(\id_\PP) \in \CCc(\PP\otimes A, B)$, because the naturality of $\GK^{AB}$ just means that
\bear
\GK^{AB}_X(F) & = &\CCc(F, \PP)\circ \GK^{AB}_\PP(\id_\PP)\ =\ (F\otimes A) \comp \UK^{AB} \ =\ \{F\}^{AB}
\eear
where we write $\{F\}$ for $(F\otimes A) \comp \UK$ not only for convenience, but also in reverence to Kleene's work and vision \cite{KleeneSC:metamathematics}.
$$\begin{tikzar}{row sep=1em,column sep=1em}
\tot \CCc(\PP,\PP)  \arrow{dddddd}[swap]{\tot\CCc(F, \PP)}  \arrow[two heads]
{rrrrrr}{\GK^{AB}_\PP}
\&\&\&\&\&\&
\CCc(\PP \otimes A, B) \arrow{dddddd}{\CCc(F\otimes A, B)}\\
\& \id_\PP \arrow[mapsto]{dddd} \arrow[mapsto]{rrrr} \&\&\&\& \UK \arrow[mapsto]{dddd} \\
\\
\\
\\
\& F \arrow[mapsto]{rrrr} \&\&\&\& \{F\}\\ 
 \tot \CCc(X,\PP)  \arrow[two heads]
 {rrrrrr}[swap]{\GK^{AB}_X} 
 \&\&\&\&\&\& \CCc(X\otimes A, B)
\end{tikzar}$$
Expressed in terms of the properties of the induced computation $\UK^{AB} \in \CCc(\PP\otimes A, B)$, the assumption that the program evaluations $\GK^{AB}_X:\tot \CCc(X,\PP) \epi \CCc(X \otimes A, B)$ are surjective functions for all $X$ means that for every computation $h\in \CCc(X\otimes A, B)$ there is an $X$-indexed family of programs $H\in \tot \CCc (X,\PP)$ such that
\bea\label{eq:yon}
  \UKK {H_x} a\  & = &\ \ h(x, a)\\[1ex] 
 \def\JPicScale{.65}\newcommand{\xhh}{\scriptstyle X}\newcommand{\aah}{\scriptstyle H}\renewcommand{\dh}{\scriptstyle B}\newcommand{\ahh}{\UK
 }\newcommand{\bhh}{\scriptstyle A} \input{PIC/uni-eval-p-x.tex}  & = &\ \  \def\JPicScale{.7}\newcommand{\aah}{\scriptstyle A}\newcommand{\xhh}{\scriptstyle X}\renewcommand{\dh}{\scriptstyle B}\newcommand{\ahh}{h} \input{PIC/uni-eval-f-x.tex} \notag
\qquad \qquad \qquad  \begin{tikzar}{row sep=1.3em,column sep=1.5em}
\& B\\   
\PP \otimes A \ar{ur}{\UK}  \\ 
 \& X \otimes A \ar{uu}[swap]{h} \ar[dashed]{ul}{H \otimes A} 
 \end{tikzar}
\eea

By the Yoneda lemma \cite[III.2]{MacLaneS:CWM}, specifying a  program evaluation $\GK^{AB}:\tot \CCc(-,\PP) \epi \CCc(- \otimes A, B)$ is thus equivalent to specifying a computation $\UK^{AB} \in \CCc(\PPp\otimes A, B)$ satisfying \eqref{eq:yon}. The task of proving the Proposition thus boils down to showing that {\it(\ref{itemuniv})}$\iff${\it(\ref{itempart})}$\Rightarrow$\eqref{eq:yon}.

\noindent Setting $X$ to be the tensor unit $I$, shows that $\UK^{AB} \in \CCc(\PP\otimes A, B)$ induced as above by $\GK^{AB}:\tot \CCc(-,\PP) \to \CCc(- \otimes A, B)$ satisfies condition {\it(\ref{itemuniv})}, and is thus a universal evaluator. To construct a partial evaluator $\SK^{ABC} \in \tot\CCc(\PP\otimes A, \PP)$ satisfying {\it(\ref{itempart})}, consider the following diagram
$$\begin{tikzar}{row sep=1em,column sep=1em}
\tot \CCc(\PP,\PP)  \arrow{dddddd}[swap]{\tot\CCc(\SK^{ABC}, \PP)}  \arrow[two heads]
{rrrrrr}{\GK^{BC}_\PP}  \arrow[two heads]
{rrrrrrdddddd}[swap]{\GK^{(AB)C}_\PP}
\&\&\&\&\&\&
\CCc(\PP \otimes B, C) \arrow{dddddd}{\CCc(\SK^{ABC}\otimes B, C)}\\
\\
\\
\\
\\
\\ 
 \tot \CCc(\PP\otimes A,\PP)  \arrow[two heads]
 {rrrrrr}[swap]{\GK^{BC}_{\PP A}} 
 \&\&\&\&\&\& \CCc(\PP \otimes A\otimes B, C)
\end{tikzar}$$
Since $\GK^{BC}_{\PP A}$ is surjective, there must exist $\SK^{ABC} \in \tot \CCc(\PP\otimes A, \PP)$ such that
\bear
\GK^{BC}_{\PP A}\left( \SK^{ABC}\right) & = & \GK^{(AB)C}_{\PP}\left( \id_\PP\right)
\eear
Fix a choice of such $\SK^{ABC} \in \tot \CCc(\PP\otimes A, \PP)$,  and chase the above diagram. Recalling that $\UK^{(AB)C} = \GK^{(AB)C}_{\PP}\left( \id_\PP\right)$ and $\UK^{BC} = \GK^{BC}_{\PP}\left( \id_\PP\right)$, and observing that $\tot \CCc\left(\SK^{ABC}, \PP\right) (\id_\PP) = \SK^{ABC} \comp \id_\PP = \SK^{ABC}$, the naturality of $\GK^{BC}$ implies that
\begin{multline*}
\left(\SK^{ABC} \otimes B\right) \comp \UK^{BC} \  = \ \CCC(\SK^{ABC}\otimes B, C) \comp \GK^{BC}_\PP(\id_\PP)\ =\\ 
=\ \GK^{BC}_{\PP A} \circ \tot \CCC(\SK^{ABC}, \PP)(\id_\PP)\ =\ \GK^{(AB)C}_\PP(\id_\PP)\ =\ \UK^{(AB)C}
\end{multline*}
Written in the bracket notation, this boils down to $$\big\{\left[G, a\right]^{ABC}\big\}^{BC}b \ =\ \big\{G\big\}^{(AB)C}(a,b)$$
which shows that $\SK^{ABC}$ satisfies {\it(\ref{itempart})}, as claimed.

Turning to the converse, suppose that universal evaluators $\UK^{AB}$ and partial evaluators $\SK^{ABC}$ are given, satisfying  {\it(\ref{itemuniv})} and  {\it(\ref{itempart})}. We show that the universal evaluators $\UK^{AB}$ then satisfy the stronger requirement \eqref{eq:yon}. Since we showed above that giving $\UK^{AB} \in \CCc(\PP\otimes A, B)$ satisfying \eqref{eq:yon} is equivalent to specifying a natural family of surjections $\GK^{AB}:\tot \CCc(-,\PP) \to \CCc(- \otimes A, B)$, this will complete the proof.

Towards the proof that {\it(\ref{itemuniv})}$\wedge${\it(\ref{itempart})}$\Rightarrow$\eqref{eq:yon}, consider an arbitrary computation $h\in \CCc(XA, B)$. Then \begin{itemize}
\item {\it (\ref{itemuniv})} gives $\widetilde H \in \CCc(\PP)$ such that $\left\{ {\widetilde H}\right\}^{(XA)B} (x,a) = h(x,a)$, and
\item {\it (\ref{itempart})} gives $H_x =\left[\widetilde H, x\right]^{XAB} \in \CCc(X, \PP)$ such that $\big\{H_x\big\}^{AB} a = h(x,a)$.
\end{itemize}
\bigskip
\[
\newcommand{\computations}{$h(x,a)$}
\newcommand{\programmed}{$\UKK{\widetilde H} {(x,a)}$}
\newcommand{\evaluated}{$\UKK{\SKK {\widetilde H} x} y$}
\newcommand{\program}{$\scriptstyle \widetilde H$}
\newcommand{\progH}{$\color{red} H$}
\newcommand{\gee}{$h$}
\newcommand{\Progtype}{}
\newcommand{\Progk}{X}
\newcommand{\Progm}{A}
\newcommand{\Progn}{B}
\newcommand{\EQLS}{=}
\newcommand{\universal}{\{\}}
\newcommand{\universalmkn}{\{\}}
\renewcommand{\partial}{[\,]}
\def\JPicScale{.3}
\input{PIC/computation-H.tex}
\]
\end{proof}

\medskip
\noindent{\bf Branching.} By extending the $\lambda$-calculus constructions as in \cite{PavlovicD:IC12}, we can extract from $\PP$ the convenient types of natural numbers, truth values, etc. E.g., if the truth values $\true$ and $\false$ are defined to be some programs for the two projections, then the role of the $\iif$-branching command can be played by the universal evaluator:
\bear
\iif (b, x, y) & = & \UKK b (x,y)\ \ =\ \  \begin{cases} x & \mbox{ if }\  b = \true \\[2em]
	y & \mbox{ if }\  b = \false
\end{cases}\hspace{3em} \newcommand{\zzero}{\scriptstyle  \false}
\newcommand{\oone}{\scriptstyle  \true}
\newcommand{\EQLS}{=}
\newcommand{\ueval}{{\scriptstyle if =  \UK}}
\def\JPicScale{.23}
\input{PIC/tfalse-3.tex}
\eear

\subsection{Examples of monoidal computer}\label{sec:examples}
Let $\SSS$ be a cartesian category and $T:\SSS\to \SSS$ a commutative monad. Then the Kleisli category $\SSS_T$ of free algebras is monoidal, with the data services induced by the cartesian structure of $\SSS$. 

The standard model of monoidal computer $\CC$ is obtained by taking $\SSS$ to be the category of finite and countable sets, and $TX = \bot +X$ to be the \emph{maybe}\/ monad, adjoining a fresh element to every set. The category $\SSS_\bot$ is the category of partial functions, and the monoidal computer $\CC\subseteq \SSS_\bot$ is the subcategory of \emph{computable}\/ partial functions:
\begin{itemize}
\item $|\CC|\  =\  \big\{A\subseteq \NNn\ |\ \exists e\in\NNn.\ \{e\}a\halts \iff a\in A\big\}$
\item $\CC(A,B)\ =\ \big\{f:A\pfn B\ |\ \exists e.\ \{e\} = f\big\}$
\end{itemize}
The category $\tot \CC$ is then the category of computable \emph{total}\/ functions. Assuming that the programs are encoded as natural numbers, the type of programs is $\PP = \NNn$; but any language containing a Turing complete set of expressions would do, \emph{mutatis mutandis}. The sequence $\{0\}, \{1\}, \{2\},\ldots$ denotes an acceptable enumeration of computable partial functions \cite[II.5]{Odifreddi:one}. The universal evaluators can be implemented as partial recursive functions; the partial evaluators are the total recursive functions, constructed in Kleene's $S^m_{n}$-theorem \cite{KleeneSC:ordinal}. 

Other commutative monads $T:\SSS\to \SSS$ induce monoidal computers in a similar way, capturing intensional computations together with the corresponding computational  effects: exceptions, nondeterminism, randomness \cite{MoggiE:monads}. Some of the familiar computational monads need to be restricted to finite support. The distribution monad must be factored modulo computational indistinguishability. A simple quantum monoidal computer can be constructed using a relative monad for finite dimensional vector spaces \cite{AltenkirchT:relative-monads}. However, in the model where the universal evaluators are quantum Turing machines, the program evaluations cannot be surjective in the usual sense, but only in the topologically enriched sense, i.e., they are dense \cite{Bernstein-Vazirani}. We do not know how to derive this model from a computational monad, albeit relative. Another interesting feature is that most computational effects induce nonstandard data services, corresponding to complementary bases, which are, of course, used in randomized, quantum, but also in nondeterministic algorithms \cite{PavlovicD:QI09,PavlovicD:Qabs12}. More examples are in \cite{PavlovicD:IC12}, but most work is still ahead.

\subsection{Encoding all types}\label{Sec-retracts}
\begin{proposition}\label{prop:retracts}
Every type $B$ in a monoidal computer is a retract of the type of programs $\PP$. More precisely, for every type $B\in |\CCc|$ there are computations
$$\begin{tikzar}{}
B \arrow[tail,bend left = 15]{rr}{e} \&\&\PP \arrow[two heads,bend left = 15]{ll}{d}
\end{tikzar}$$
such that $\enc^B$ is a map, and $\enc^B\comp \dec^B = \id_B$. We often call $\enc^B$ the \emph{encoding}\/ of $B$ and $\dec^B\in \CCc(\PP,B)$ is the corresponding \emph{decoding}.
\end{proposition}

\begin{remark}	
Note that there is no claim that either $\enc^B$ or $\dec^B$ is unique. Indeed, in nondegenerate monoidal computers, each type $B$ has many different encoding pairs $\enc^B, \dec^B$. However, once such a pair is chosen, the fact that $\enc^B$ is total and single-valued means that it assigns  a unique program code to each element of $B$. The fact that $\dec^B$ is not total means that some programs in $\PPp$ may not correspond to elements of $B$. 

Since Prop.~\ref{prop:retracts} says that the program evaluations make every type into a retract of $\PP$, and Prop.~\ref{prop:us} reduced the structure of monoidal computer to the evaluators for all types, it is natural to ask if the evaluators of all types can be reduced to the evaluators over the type $\PP$ of programs. Can all of the structure of a monoidal computer be derived from the structure of the type $\PP$ of programs? E.g., can the program evaluations be \emph{"uniformized"} by always encoding the input data of all types in $\PP$, performing the evaluations to get the outputs in $\PP$, and then decoding the outputs back to the originally given types? Can the type structure and the evaluation structure of a monoidal computer be reconstructed by unfolding the structure of $\PP$, as it is the case in models of $\lambda$-calculus? Is monoidal computer yet another categorical view of a partial applicative structure? 

The answer to all these question is positive \emph{just}\/ in the degenerate case of an essentially extensional monoidal computer. If the type structure of monoidal computer can be faithfully encoded in $\PP$, then there is a retract of $\PP$ which supports an extensional model of computation, i.e. allows assigning a unique program to each computation. 

If all evaluators can be derived by decoding the evaluators with the output type $\PP$, and if the decoding preserves the original evaluators on $\PP$, then all computation representable in monoidal computer must be provably total and single valued: it degerates into a cartesian closed category derived from a $C$-monoid. For details see  \cite[I.15-I.17]{
	Lambek-Scott:book}, and the references therein.
\end{remark}

\begin{proof} The claimed retraction $\begin{tikzar}{}
B \arrow[tail,bend left = 15]{r}{e} \&\PP \arrow[two heads,bend left = 15]{l}{d}
\end{tikzar}$ can be found using the following diagram: 

$$\begin{tikzar}{row sep=1em,column sep=1em}
\tot \CCc(\PP,\PP)  \arrow{dddddd}[swap]{\tot\CCc(\enc^B, \PP)}  \arrow[two heads]{rrrrrr}{\GK^{IB}_\PP}
\&\&\&\&\&\&
\CCc(\PP, B) \arrow{dddddd}{\CCc(\enc^B, B)}\\
\& \id_\PP \arrow[mapsto]{dddd} \arrow[mapsto]{rrrr} \&\&\&\& \dec^B \arrow[mapsto]{dddd} \\
\\
\\
\\
\& \enc^B \arrow[mapsto]{rrrr} \&\&\&\& \id_B\\ 
 \tot \CCc(B,\PP)  \arrow[two heads]{rrrrrr}[swap]{\GK^{IB}_B} 
 \&\&\&\&\&\& \CCc(B, B)
\end{tikzar}$$
While $\dec^B$ is defined to be \emph{the}\/ image of $\id_\PP$ along $\GK_\PP^{I\PP}$, $\enc^B$ is defined to be \emph{any}\/ inverse image of $\id_B$ along $\gamma^{I\PP}_B$, which must exist because $\gamma^{I\PP}_B$ is a surjection. 
so $$\GK_\PP^{I\PP}(\id_\PP) = \dec^B \qquad \mbox{and} \qquad  \GK_B^{I\PP}(\enc^B) = \id_B$$  The fact that $\enc^B\comp \dec^B = \id_B$ follows from the naturality of $\GK^{IB}$, which implies that the square in the diagram commutes, and therefore
\begin{multline*}
\enc^B \comp \dec^B\  = \ \CCc(\enc^B, B) \left(\dec^B\right)\  =\  \CCc(\enc^B, B) \circ \GK^{I\PP}_\PP\left(\id_B\right)\  =\\  =\ \GK^{I\PP}_B\circ \tot \CCc(\enc^B, \PP)\left(\id_B\right)\  =\  \GK^{I\PP}_B\left(\enc^B\right)\   =\  \id_B
\end{multline*}
\end{proof}

\begin{remark}
In \cite{PavlovicD:IC12} we only considered the \emph{basic}\/ monoidal computer, where all types are powers of $\PP$. In the standard model, programs are encoded as natural numbers, and all data are tuples of natural numbers, which can be recursively encoded as natural numbers. Prop.~\ref{prop:retracts} says that this must be the case in every computer. 
\end{remark}

\subsection{The Fundamental Theorem of Computability}\label{Sec-fund}
In this section we show that every monoidal computer validates the claim of Kleene's fundamental result, which he called the Second Recursion Theorem \cite{KleeneSC:ordinal,MoschovakisY:Kleene}.

\begin{theorem}\label{thm:kleene}
In every monoidal computer $\CCc$, for every computation $g\in \CCc(\PP A,B)$ there is a program $\Gamma\in \CCc(\PP)$ such that

\smallskip
\newcommand{\lhs}{g(\Gamma, a)}
\newcommand{\rhs}{\UKK \Gamma a}
\newcommand{\inputt}{\scriptstyle A}
\newcommand{\universal}{$\UK$}
\newcommand{\gee}{\it g}
\newcommand{\nameslang}{\scriptstyle B}
\newcommand{\program}{$\scriptstyle \Gamma$}
\newcommand{\EQLS}{\large =}
\newcommand{\progtype}{\scriptstyle \PP}
\def\JPicScale{.3}
\begin{center}
\input{PIC/kleene.tex}
\end{center}
We call $\Gamma$ \emph{Kleene's fixed program}\/ of $g$.
\end{theorem}

\begin{proof}
Let $G$ be a program such that
\medskip
\newcommand{\inputt}{\scriptstyle A}
\newcommand{\universal}{$\UK$}
\newcommand{\gee}{\it g}
\newcommand{\nameslang}{\scriptstyle B}
\newcommand{\program}{$\scriptstyle G$}
\newcommand{\parteval}{\scriptstyle \SK}
\newcommand{\progtype}{\scriptstyle \PP}
\newcommand{\EQLS}{\large =}
\newcommand{\lhs}{g\big(\SKK p p, a\big)}
\newcommand{\rhs}{ \{G\}(p,a)}
\def\JPicScale{.3}
\begin{center}
\input{PIC/kleene-1.tex}
\end{center}
A Kleene fixed program $\Gamma$ can now be constructed by evaluating $G$ on itself, i.e. as $\Gamma = \SKK G G$, because

\smallskip
\newcommand{\oone}{g(\Gamma, a)}
\newcommand{\ttwo}{g\big([G, G], a\big)}
\newcommand{\tthree}{\{G\}(G, a)}
\newcommand{\ffour}{\big\{[G, G]\big\}a}
\newcommand{\ffive}{\left\{\Gamma \right\}a}
\renewcommand{\parteval}{\scriptstyle \SK}
\newcommand{\progH}{$\color{red} \Gamma$}
\def\JPicScale{.25}
\begin{center}
\input{PIC/kleene-2.tex}
\end{center}
\vspace{-2\baselineskip}
\end{proof}

The Fundamental Theorem allows constructing convenient representations of integers, arithmetic, primitive recursion, and unbounded search, and thus proving that monoidal computer is Turing complete. In \cite{PavlovicD:IC12}, this was done by using the $\lambda$-calculus constructions. In the next section, we provide yet another proof, by implementing Turing machines.

%% file: PIC/uni-eval-f.tex
\ifx\JPicScale\undefined\def\JPicScale{1}\fi
\psset{unit=\JPicScale mm}
\psset{linewidth=0.3,dotsep=1,hatchwidth=0.3,hatchsep=1.5,shadowsize=1,dimen=middle}
\psset{dotsize=0.7 2.5,dotscale=1 1,fillcolor=black}
\psset{arrowsize=1 2,arrowlength=1,arrowinset=0.25,tbarsize=0.7 5,bracketlength=0.15,rbracketlength=0.15}
\begin{pspicture}(0,0)(13.75,20)
\rput(5,20){$\dh$}
\psline(5,-8.25)(5,-1.25)
\psline(5,8.75)(5,15.5)
\rput(5,3.75){$\ahh$}
\rput(5,-13.12){$\aah$}
\psline(-3.76,4.99)
(-3.75,-1.25)
(13.75,-1.26)
(13.75,8.74)
(-3.76,8.74)
(13.75,8.74)
(-3.76,8.74)(-3.76,4.99)
\end{pspicture}

%% file: PIC/uni-eval-p.tex
\ifx\JPicScale\undefined\def\JPicScale{1}\fi
\psset{unit=\JPicScale mm}
\psset{linewidth=0.3,dotsep=1,hatchwidth=0.3,hatchsep=1.5,shadowsize=1,dimen=middle}
\psset{dotsize=0.7 2.5,dotscale=1 1,fillcolor=black}
\psset{arrowsize=1 2,arrowlength=1,arrowinset=0.25,tbarsize=0.7 5,bracketlength=0.15,rbracketlength=0.15}
\begin{pspicture}(0,0)(23.12,20)
\rput(16.88,20){$\dh$}
\psline(6.88,-0.62)(6.88,3.12)
\psline(16.88,8.12)(16.88,15.62)
\rput(16.88,3.12){$\ahh$}
\rput(5,-2.5){$\aah$}
\psline(17.5,-8.11)(17.5,-1.88)
\rput(17.5,-12.5){$\bhh$}
\pspolygon(1.88,4.36)
(11.86,-5.62)
(1.88,-5.62)(1.88,4.36)
\psline(1.88,8.12)
(11.88,-1.87)
(23.12,-1.88)
(23.12,8.12)
(5.62,8.12)
(23.12,8.12)
(5.62,8.12)(1.88,8.12)
\end{pspicture}

%% file: PIC/uni-eval-uMN.tex
\ifx\JPicScale\undefined\def\JPicScale{1}\fi
\psset{unit=\JPicScale mm}
\psset{linewidth=0.3,dotsep=1,hatchwidth=0.3,hatchsep=1.5,shadowsize=1,dimen=middle}
\psset{dotsize=0.7 2.5,dotscale=1 1,fillcolor=black}
\psset{arrowsize=1 2,arrowlength=1,arrowinset=0.25,tbarsize=0.7 5,bracketlength=0.15,rbracketlength=0.15}
\begin{pspicture}(0,0)(25.62,14.38)
\rput(4.38,-11.25){$\dh$}
\psline(3.75,-7.5)(3.75,1.88)
\psline(14.38,-7.5)(14.38,-3.12)
\rput(14.38,-11.25){$\bhh$}
\rput(16.25,1.88){$\ahh$}
\psline(21.25,-7.5)(21.25,-3.12)
\rput(21.25,-11.25){$\chh$}
\psline(15.62,6.24)(15.62,11.25)
\rput(15.62,14.38){$\Lhh$}
\psline(-0.62,6.25)
(8.75,-3.12)
(25.62,-3.12)
(25.62,6.25)
(5.62,6.26)
(3.12,6.26)
(1.25,6.25)(-0.62,6.25)
\end{pspicture}

%% file: PIC/uni-eval-s.tex
\ifx\JPicScale\undefined\def\JPicScale{1}\fi
\psset{unit=\JPicScale mm}
\psset{linewidth=0.3,dotsep=1,hatchwidth=0.3,hatchsep=1.5,shadowsize=1,dimen=middle}
\psset{dotsize=0.7 2.5,dotscale=1 1,fillcolor=black}
\psset{arrowsize=1 2,arrowlength=1,arrowinset=0.25,tbarsize=0.7 5,bracketlength=0.15,rbracketlength=0.15}
\begin{pspicture}(0,0)(30.01,13.12)
\psline(2.5,-11.88)(2.5,-3.1)
\psline(18.13,-2.5)(18.13,1.25)
\psline(15,-11.88)(15,-7.5)
\rput(15,-15.62){$\bhh$}
\rput(24.37,1.25){$\ahh$}
\rput(11.25,-3.12){$\shh$}
\psline(26.87,-12.5)(26.87,-3.75)
\rput(26.87,-15.62){$\chh$}
\rput(2.5,-15.62){$\dh$}
\psline(26.87,5.62)(26.87,10.62)
\rput(26.87,13.12){$\Lhh$}
\psline(13.76,5.62)
(23.13,-3.75)
(30.01,-3.75)
(30.01,5.62)
(26.87,5.62)
(16.87,5.62)
(15,5.62)(13.76,5.62)
\pscustom[]{\psline(-2.5,1.88)(6.88,-7.5)
\psline(6.88,-7.5)(23.13,-7.5)
\psline(23.13,-7.5)(16.86,-1.25)
\psline(16.86,-1.25)(13.74,1.88)
\psline(13.74,1.88)(-2.5,1.88)
\psline(-2.5,1.88)(-1.25,1.88)
\psbezier(-1.25,1.88)(-1.25,1.88)(-1.25,1.88)
}
\end{pspicture}

%% file: PIC/computation.tex
\ifx\JPicScale\undefined\def\JPicScale{1}\fi
\psset{unit=\JPicScale mm}
\psset{linewidth=0.3,dotsep=1,hatchwidth=0.3,hatchsep=1.5,shadowsize=1,dimen=middle}
\psset{dotsize=0.7 2.5,dotscale=1 1,fillcolor=black}
\psset{arrowsize=1 2,arrowlength=1,arrowinset=0.25,tbarsize=0.7 5,bracketlength=0.15,rbracketlength=0.15}
\begin{pspicture}(0,0)(224.38,100)
\psline[linewidth=0.75](106.25,50)(106.25,5)
\psline[linewidth=0.75](124.38,50)(124.38,5)
\psline[linewidth=0.75](96.24,50)(76.24,70)
\psline[linewidth=0.75](76.24,70)(133.75,70)
\psline[linewidth=0.75](96.24,50)(133.12,50)
\rput(150,59.38){\EQLS}
\psline[linewidth=0.75](133.75,70)(133.75,50)
\psline[linewidth=0.75](117.5,80)(117.5,70)
\psline[linewidth=0.75](85.62,60)(85.62,39.38)
\psline[linewidth=0.75](178.12,43.12)(178.12,25.62)
\psline[linewidth=0.75](198.74,33.12)(198.75,5)
\psline[linewidth=0.75](168.12,53.12)(189.38,53.12)
\psline[linewidth=0.75](209.38,33.12)(188.12,33.12)
\psline[linewidth=0.75](210,50)(190,70)
\psline[linewidth=0.75](189.38,53.12)(209.38,33.12)
\psline[linewidth=0.75](190,70)(224.38,70)
\psline[linewidth=0.75](198.12,61.88)(198.12,45)
\psline[linewidth=0.75](210,50)(224.38,50)
\psline[linewidth=0.75](224.38,70)(224.38,50)
\psline[linewidth=0.75](168.12,53.12)(188.12,33.12)
\psline[linewidth=0.75](216.25,80)(216.25,69.38)
\rput[r](196.25,57.5){$\Progtype$}
\psline[linewidth=0.75](216.25,50)(216.25,5)
\rput(188.74,43.12){$\partial$}
\rput(215,60){$\universal$}
\rput(106.25,-1.25){$\Progk$}
\rput(123.75,-1.25){$\Progm$}
\rput(116.88,85.62){$\Progn$}
\rput(216.25,85.62){$\Progn$}
\rput(216.25,0){$\Progm$}
\rput(198.75,0){$\Progk$}
\rput(111.25,60){$\universalmkn$}
\psline[linewidth=0.75](167.5,36.25)(167.5,11.25)
\psline[linewidth=0.75](167.5,11.25)(192.5,11.25)
\rput(175,18.12){\program}
\psline[linewidth=0.75](167.5,36.25)(192.5,11.25)
\psline[linewidth=0.75](73.12,26.88)(98.12,26.88)
\rput(80.62,33.75){\program}
\psline[linewidth=0.75](73.12,51.88)(98.12,26.88)
\psline[linewidth=0.75](73.12,51.88)(73.12,26.88)
\rput[r](176.88,38.75){$\Progtype$}
\rput[r](83.12,55.62){$\Progtype$}
\psline[linewidth=0.75](10,50)(10,5)
\psline[linewidth=0.75](30,50)(30,5)
\psline[linewidth=0.75](-0,50)(0,70)
\psline[linewidth=0.75](0,70)(40,70)
\psline[linewidth=0.75](-0,50)(40,50)
\psline[linewidth=0.75](40,70)(40,50)
\psline[linewidth=0.75](20,80)(20,70)
\rput(10,0){$\Progk$}
\rput(30,0){$\Progm$}
\rput(20,85){$\Progn$}
\rput(55,59.38){\EQLS}
\rput(20,60){\gee}
\rput(20,100){\computations}
\rput(110,100){\programmed}
\rput(195,100){\evaluated}
\rput(55,100){\EQLS}
\rput(150,100){\EQLS}
\end{pspicture}

%% file: PIC/uni-eval-p-x.tex
\ifx\JPicScale\undefined\def\JPicScale{1}\fi
\psset{unit=\JPicScale mm}
\psset{linewidth=0.3,dotsep=1,hatchwidth=0.3,hatchsep=1.5,shadowsize=1,dimen=middle}
\psset{dotsize=0.7 2.5,dotscale=1 1,fillcolor=black}
\psset{arrowsize=1 2,arrowlength=1,arrowinset=0.25,tbarsize=0.7 5,bracketlength=0.15,rbracketlength=0.15}
\begin{pspicture}(0,0)(23.12,17.5)
\rput(16.88,17.5){$\dh$}
\psline(6.88,-1.87)(6.88,1.87)
\psline(16.88,6.87)(16.88,14.37)
\rput(16.88,1.87){$\ahh$}
\rput(5,-3.75){$\aah$}
\psline(17.5,-12.5)(17.5,-3.13)
\rput(17.5,-15.62){$\bhh$}
\pspolygon(1.88,3.11)
(11.86,-6.87)
(1.88,-6.87)(1.88,3.11)
\psline(1.88,6.87)
(11.88,-3.12)
(23.12,-3.13)
(23.12,6.87)
(5.62,6.87)
(23.12,6.87)
(5.62,6.87)(1.88,6.87)
\psline(6.25,-12.5)(6.25,-6.87)
\rput(6.25,-15.63){$\xhh$}
\end{pspicture}

%% file: PIC/uni-eval-f-x.tex
\ifx\JPicScale\undefined\def\JPicScale{1}\fi
\psset{unit=\JPicScale mm}
\psset{linewidth=0.3,dotsep=1,hatchwidth=0.3,hatchsep=1.5,shadowsize=1,dimen=middle}
\psset{dotsize=0.7 2.5,dotscale=1 1,fillcolor=black}
\psset{arrowsize=1 2,arrowlength=1,arrowinset=0.25,tbarsize=0.7 5,bracketlength=0.15,rbracketlength=0.15}
\begin{pspicture}(0,0)(13.75,16.25)
\rput(5,16.25){$\dh$}
\psline(8.75,-10.75)(8.75,-3.75)
\psline(5,6.25)(5,13)
\rput(5,1.25){$\ahh$}
\rput(8.75,-14.38){$\aah$}
\psline(-3.76,2.49)
(-3.75,-3.75)
(13.75,-3.76)
(13.75,6.24)
(-3.76,6.24)
(13.75,6.24)
(-3.76,6.24)(-3.76,2.49)
\psline(0,-11.25)(0,-3.75)
\rput(0,-14.38){$\xhh$}
\end{pspicture}

%% file: PIC/computation-H.tex
\ifx\JPicScale\undefined\def\JPicScale{1}\fi
\psset{unit=\JPicScale mm}
\psset{linewidth=0.3,dotsep=1,hatchwidth=0.3,hatchsep=1.5,shadowsize=1,dimen=middle}
\psset{dotsize=0.7 2.5,dotscale=1 1,fillcolor=black}
\psset{arrowsize=1 2,arrowlength=1,arrowinset=0.25,tbarsize=0.7 5,bracketlength=0.15,rbracketlength=0.15}
\begin{pspicture}(0,0)(265,120)
\psline[linewidth=0.75](106.25,50)(106,0)
\psline[linewidth=0.75](125,50)(125,0)
\psline[linewidth=0.75](96.24,50)(76.24,70)
\psline[linewidth=0.75](76.24,70)(133.75,70)
\psline[linewidth=0.75](96.24,50)(133.12,50)
\rput(150,59.38){\EQLS}
\psline[linewidth=0.75](133.75,70)(133.75,50)
\psline[linewidth=0.75](117.5,88.75)(117.5,70)
\psline[linewidth=0.75](85.62,60)(85.62,39.38)
\psline[linewidth=0.75](178,35)(178.12,25.62)
\psline[linewidth=0.75](198,25)(198,0)
\psline[linewidth=0.75](168.12,45.12)(189.38,45.12)
\psline[linewidth=0.75](209.38,25.12)(188.12,25.12)
\psline[linewidth=0.75](210,50)(190,70)
\psline[linewidth=0.75](189.38,45.12)(209.38,25.12)
\psline[linewidth=0.75](190,70)(224.38,70)
\psline[linewidth=0.75](198.12,61.88)(198,37)
\psline[linewidth=0.75](210,50)(224.38,50)
\psline[linewidth=0.75](224.38,70)(224.38,50)
\psline[linewidth=0.75](168.12,45.12)(188.12,25.12)
\psline[linewidth=0.75](216.25,88.12)(216.25,69.38)
\psline[linewidth=0.75](216.25,50)(216,0)
\rput(188.74,35.12){$\partial$}
\rput(215,60){$\universal$}
\rput(106.25,-6.25){$\Progk$}
\rput(123.75,-6.25){$\Progm$}
\rput(116.88,95.62){$\Progn$}
\rput(216.25,95.62){$\Progn$}
\rput(216.25,-5){$\Progm$}
\rput(198.75,-5){$\Progk$}
\rput(111.25,60){$\universalmkn$}
\psline[linewidth=0.75](167.5,36.25)(167.5,11.25)
\psline[linewidth=0.75](167.5,11.25)(192.5,11.25)
\rput(175,18.12){\program}
\psline[linewidth=0.75](167.5,36.25)(192.5,11.25)
\psline[linewidth=0.75](73.12,26.88)(98.12,26.88)
\rput(80.62,33.75){\program}
\psline[linewidth=0.75](73.12,51.88)(98.12,26.88)
\psline[linewidth=0.75](73.12,51.88)(73.12,26.88)
\psline[linewidth=0.75](10,50)(10,0)
\psline[linewidth=0.75](29,50)(29,0)
\psline[linewidth=0.75](-0,50)(0,70)
\psline[linewidth=0.75](0,70)(39.38,70)
\psline[linewidth=0.75](-0,50)(39.38,50)
\psline[linewidth=0.75](39.38,70)(39.38,50)
\psline[linewidth=0.75](18.12,88.75)(18.12,70)
\rput(10,-6.25){$\Progk$}
\rput(28.75,-5.62){$\Progm$}
\rput(18.12,95.62){$\Progn$}
\rput(55,59.38){\EQLS}
\rput(20,60){\gee}
\rput(17.5,120){\computations}
\rput(110,120){\programmed}
\rput(195,120){\evaluated}
\rput(55,120){\EQLS}
\rput(150,120){\EQLS}
\newrgbcolor{userLineColour}{1 0.2 0.2}
\psline[linewidth=0.75,linecolor=userLineColour](164,75)(164,8)
\newrgbcolor{userLineColour}{1 0.2 0.2}
\psline[linewidth=0.75,linecolor=userLineColour](164,75)(212,27)
\newrgbcolor{userLineColour}{1 0.2 0.2}
\psline[linewidth=0.75,linecolor=userLineColour](164,8)(212,8)
\newrgbcolor{userLineColour}{1 0.2 0.2}
\psline[linewidth=0.75,linecolor=userLineColour](212,27)(212,8)
\rput(205,15){\progH}
\rput(265,10){}
\end{pspicture}

%% file: PIC/tfalse-3.tex
\ifx\JPicScale\undefined\def\JPicScale{1}\fi
\psset{unit=\JPicScale mm}
\psset{linewidth=0.3,dotsep=1,hatchwidth=0.3,hatchsep=1.5,shadowsize=1,dimen=middle}
\psset{dotsize=0.7 2.5,dotscale=1 1,fillcolor=black}
\psset{arrowsize=1 2,arrowlength=1,arrowinset=0.25,tbarsize=0.7 5,bracketlength=0.15,rbracketlength=0.15}
\begin{pspicture}(0,0)(130,55)
\psline[linewidth=0.75](120,-5)(120,-50)
\psline[linewidth=0.75](100,-25)(100,-50)
\newrgbcolor{userFillColour}{0.2 0 0.2}
\rput{0}(100,-25){\psellipse[fillcolor=userFillColour,fillstyle=solid](0,0)(2.58,-2.58)}
\psline[linewidth=0.6](90,-15)(130,-15)
\psline[linewidth=0.6](130,-35)(130,-15)
\psline[linewidth=0.6](90,-35)(130,-35)
\psline[linewidth=0.6](90,-15)(90,-35)
\rput(70,-25){\EQLS}
\psline[linewidth=0.75](0,45)(50,45)
\psline[linewidth=0.75](50,25)(50,45)
\psline[linewidth=0.75](35,55)(35,45)
\psline[linewidth=0.75](20,25)(50,25)
\psline[linewidth=0.75](45,25)(45,10)
\psline[linewidth=0.75](10,35.01)(10,25)
\psline[linewidth=0.75](0,45)(20,25)
\psline[linewidth=0.75](0,15)(20,15)
\psline[linewidth=0.75](0,15)(0,35)
\psline[linewidth=0.75](0,35)(20,15)
\psline[linewidth=0.75](25,25)(25,10)
\rput(70,35){\EQLS}
\psline[linewidth=0.75](0,-15)(50,-15)
\psline[linewidth=0.75](50,-35)(50,-15)
\psline[linewidth=0.75](35,-5)(35,-15)
\psline[linewidth=0.75](20,-35)(50,-35)
\psline[linewidth=0.75](45,-35)(45,-50)
\psline[linewidth=0.75](10,-24.99)(10,-35)
\psline[linewidth=0.75](0,-15)(20,-35)
\psline[linewidth=0.75](0,-45)(20,-45)
\psline[linewidth=0.75](0,-45)(0,-25)
\psline[linewidth=0.75](0,-25)(20,-45)
\psline[linewidth=0.75](25,-35)(25,-50)
\rput(5,20){$\oone$}
\rput(6.25,-38.75){$\zzero$}
\psline[linewidth=0.75](100,55)(100,10)
\psline[linewidth=0.75](120,35.01)(120,10)
\newrgbcolor{userFillColour}{0.2 0 0.2}
\rput{0}(120,35){\psellipse[fillcolor=userFillColour,fillstyle=solid](0,0)(2.58,-2.58)}
\psline[linewidth=0.6](90,45)(130,45)
\psline[linewidth=0.6](130,25)(130,45)
\psline[linewidth=0.6](90,25)(130,25)
\psline[linewidth=0.6](90,45)(90,25)
\rput(30,35){$\ueval$}
\rput(30,-25){$\ueval$}
\end{pspicture}

%% file: PIC/kleene.tex
\ifx\JPicScale\undefined\def\JPicScale{1}\fi
\psset{unit=\JPicScale mm}
\psset{linewidth=0.3,dotsep=1,hatchwidth=0.3,hatchsep=1.5,shadowsize=1,dimen=middle}
\psset{dotsize=0.7 2.5,dotscale=1 1,fillcolor=black}
\psset{arrowsize=1 2,arrowlength=1,arrowinset=0.25,tbarsize=0.7 5,bracketlength=0.15,rbracketlength=0.15}
\begin{pspicture}(0,0)(155,102.5)
\psline[linewidth=0.75](95,70)(155,70)
\psline[linewidth=0.75](95,35.62)(95,10.62)
\psline[linewidth=0.75](155,45)(155,70)
\psline[linewidth=0.75](135,80)(135,70)
\rput(135,87.5){$\nameslang$}
\rput(135,-1.88){$\inputt$}
\psline[linewidth=0.75](120,45)(155,45)
\psline[linewidth=0.75](135,45)(135,4.38)
\psline[linewidth=0.75](95,10.62)(120,10.62)
\rput(133.12,56.88){\universal}
\rput(102.5,17.5){\program}
\rput(80,51.88){\EQLS}
\psline[linewidth=0.75](95,35.62)(120,10.62)
\psline[linewidth=0.75](107.5,57.5)(107.5,23.75)
\psline[linewidth=0.75](95,70)(120,45)
\psline[linewidth=0.75](15,70)(65,70)
\psline[linewidth=0.75](12.5,35.62)(12.5,10.62)
\psline[linewidth=0.75](65,45)(65,70)
\psline[linewidth=0.75](40,80)(40,70)
\rput(40,87.5){$\nameslang$}
\rput(55,-1.88){$\inputt$}
\psline[linewidth=0.75](15,45)(65,45)
\psline[linewidth=0.75](55,45)(55,4.38)
\psline[linewidth=0.75](12.5,10.62)(37.5,10.62)
\rput(20,17.5){\program}
\psline[linewidth=0.75](12.5,35.62)(37.5,10.62)
\psline[linewidth=0.75](25,45)(25,23.75)
\psline[linewidth=0.75](15,70)(15,45)
\rput(40,56.88){$\gee$}
\rput[l](26.88,33.12){$\progtype$}
\rput[l](109.38,35){$\progtype$}
\rput(40,102.5){$\lhs$}
\rput(135,102.5){$\rhs$}
\rput(80,102.5){\EQLS}
\end{pspicture}

%% file: PIC/kleene-1.tex
\ifx\JPicScale\undefined\def\JPicScale{1}\fi
\psset{unit=\JPicScale mm}
\psset{linewidth=0.3,dotsep=1,hatchwidth=0.3,hatchsep=1.5,shadowsize=1,dimen=middle}
\psset{dotsize=0.7 2.5,dotscale=1 1,fillcolor=black}
\psset{arrowsize=1 2,arrowlength=1,arrowinset=0.25,tbarsize=0.7 5,bracketlength=0.15,rbracketlength=0.15}
\begin{pspicture}(0,0)(173.75,132.5)
\psline[linewidth=0.75](113.75,94.37)(173.75,94.37)
\psline[linewidth=0.75](113.75,60)(113.75,35)
\psline[linewidth=0.75](173.75,69.37)(173.75,94.37)
\psline[linewidth=0.75](153.75,106.25)(153.75,94.37)
\rput(153.75,112.5){$\nameslang$}
\rput(165.62,-1.25){$\inputt$}
\psline[linewidth=0.75](138.75,69.37)(173.75,69.37)
\psline[linewidth=0.75](165,69.37)(165,3.75)
\psline[linewidth=0.75](113.75,35)(138.75,35)
\rput(151.87,81.25){\universal}
\rput(121.25,41.87){\program}
\rput(98.75,76.25){\EQLS}
\psline[linewidth=0.75](113.75,60)(138.75,35)
\psline[linewidth=0.75](126.25,81.88)(126.25,48.12)
\psline[linewidth=0.75](113.75,94.37)(138.75,69.37)
\psline[linewidth=0.75](33.75,94.38)(83.75,94.37)
\psline[linewidth=0.75](83.75,69.37)(83.75,94.37)
\psline[linewidth=0.75](58.75,106.25)(58.75,94.37)
\rput(58.75,112.5){$\nameslang$}
\rput(74.38,-1.25){$\inputt$}
\psline[linewidth=0.75](33.75,69.38)(83.75,69.37)
\psline[linewidth=0.75](73.75,69.37)(73.75,3.75)
\psline[linewidth=0.75](43.75,69.38)(43.75,48.12)
\psline[linewidth=0.75](33.75,94.38)(33.75,69.38)
\rput(58.75,81.25){$\gee$}
\psline[linewidth=0.75](13.75,56.24)(35,56.24)
\psline[linewidth=0.75](55,36.25)(33.75,36.25)
\psline[linewidth=0.75](35,56.24)(55,36.25)
\psline[linewidth=0.75](13.75,56.24)(33.75,36.25)
\psline[linewidth=0.75](34.38,14.38)(34.38,3.75)
\newrgbcolor{userFillColour}{0.2 0 0.2}
\rput{0}(34.37,14.37){\psellipse[fillcolor=userFillColour,fillstyle=solid](0,0)(1.41,-1.41)}
\psline[linewidth=0.75](45,36.25)(45,25)
\psline[linewidth=0.75](24.38,45)(24.38,23.76)
\psline[linewidth=0.75](45,25)(34.38,14.38)
\psline[linewidth=0.75](24.38,23.76)(34.38,13.75)
\psline[linewidth=0.75](145,69.37)(145,3.75)
\rput(145,-0.62){$\progtype$}
\rput(34.38,-1.25){$\progtype$}
\rput(35,45){$\parteval$}
\rput[l](46.88,60){$\progtype$}
\rput[l](128.75,60.62){$\progtype$}
\rput(100,132.5){\EQLS}
\rput(60,132.5){$\lhs$}
\rput(150,132.5){$\rhs$}
\end{pspicture}

%% file: PIC/kleene-2.tex
\ifx\JPicScale\undefined\def\JPicScale{1}\fi
\psset{unit=\JPicScale mm}
\psset{linewidth=0.3,dotsep=1,hatchwidth=0.3,hatchsep=1.5,shadowsize=1,dimen=middle}
\psset{dotsize=0.7 2.5,dotscale=1 1,fillcolor=black}
\psset{arrowsize=1 2,arrowlength=1,arrowinset=0.25,tbarsize=0.7 5,bracketlength=0.15,rbracketlength=0.15}
\begin{pspicture}(0,0)(400.25,150)
\psline[linewidth=0.75](147.5,113.12)(207.5,113.12)
\psline[linewidth=0.75](147.5,78.75)(147.5,53.75)
\psline[linewidth=0.75](207.5,88.12)(207.5,113.12)
\psline[linewidth=0.75](187.5,125)(187.5,113.12)
\rput(187.5,131.25){$\nameslang$}
\rput(198.75,1.25){$\inputt$}
\psline[linewidth=0.75](172.5,88.12)(207.5,88.12)
\psline[linewidth=0.75](198.75,88.12)(198.75,7.5)
\psline[linewidth=0.75](147.5,53.75)(172.5,53.75)
\rput(185.62,100){\universal}
\rput(155,60.62){\program}
\rput(132.5,95){\EQLS}
\psline[linewidth=0.75](147.5,78.75)(172.5,53.75)
\psline[linewidth=0.75](160,100.62)(160,66.88)
\psline[linewidth=0.75](147.5,113.12)(172.5,88.12)
\psline[linewidth=0.75](67.5,113.12)(117.5,113.12)
\psline[linewidth=0.75](117.5,88.12)(117.5,113.12)
\psline[linewidth=0.75](92.5,125)(92.5,113.12)
\rput(92.5,131.25){$\nameslang$}
\rput(107.5,1.25){$\inputt$}
\psline[linewidth=0.75](67.5,88.12)(117.5,88.12)
\psline[linewidth=0.75](107.5,88.12)(107.5,6.88)
\psline[linewidth=0.75](77.5,88.12)(77.5,66.87)
\psline[linewidth=0.75](67.5,113.12)(67.5,88.12)
\rput(92.5,100){$\gee$}
\psline[linewidth=0.75](47.5,75)(68.75,75)
\psline[linewidth=0.75](88.75,55)(67.5,55)
\psline[linewidth=0.75](68.75,75)(88.75,55)
\psline[linewidth=0.75](47.5,75)(67.5,55)
\psline[linewidth=0.75](68.12,33.12)(68.12,22.5)
\newrgbcolor{userFillColour}{0.2 0 0.2}
\rput{0}(68.13,33.13){\psellipse[fillcolor=userFillColour,fillstyle=solid](0,0)(1.41,-1.41)}
\psline[linewidth=0.75](78.75,55)(78.75,43.75)
\psline[linewidth=0.75](58.12,63.75)(58.12,42.5)
\psline[linewidth=0.75](78.75,43.75)(68.12,33.12)
\psline[linewidth=0.75](58.12,42.5)(68.12,32.5)
\psline[linewidth=0.75](178.75,88.12)(178.75,22.5)
\psline[linewidth=0.75](57.5,33.75)(57.5,8.75)
\psline[linewidth=0.75](57.5,8.75)(82.5,8.75)
\rput(65,15.62){\program}
\psline[linewidth=0.75](57.5,33.75)(82.5,8.75)
\psline[linewidth=0.75](166.88,34.38)(166.88,9.38)
\psline[linewidth=0.75](166.88,9.38)(191.88,9.38)
\rput(174.38,16.25){\program}
\psline[linewidth=0.75](166.88,34.38)(191.88,9.38)
\psline[linewidth=0.75](238.12,112.5)(298.12,112.5)
\psline[linewidth=0.75](298.12,87.5)(298.12,112.5)
\psline[linewidth=0.75](278.12,124.38)(278.12,112.5)
\rput(278.12,130.62){$\nameslang$}
\rput(289.38,0.62){$\inputt$}
\psline[linewidth=0.75](263.12,87.5)(298.12,87.5)
\psline[linewidth=0.75](289.38,87.5)(289.38,6.88)
\rput(276.24,99.38){\universal}
\rput(223.12,94.38){\EQLS}
\psline[linewidth=0.75](238.12,112.5)(263.12,87.5)
\psline[linewidth=0.75](260,35)(260,24.38)
\newrgbcolor{userFillColour}{0.2 0 0.2}
\rput{0}(260,35){\psellipse[fillcolor=userFillColour,fillstyle=solid](0,0)(1.41,-1.41)}
\psline[linewidth=0.75](270.62,87.5)(270.62,45.62)
\psline[linewidth=0.75](250,100.62)(250,44.38)
\psline[linewidth=0.75](270.62,45.62)(260,35)
\psline[linewidth=0.75](250,44.38)(260,34.38)
\psline[linewidth=0.75](249.38,35.62)(249.38,10.62)
\psline[linewidth=0.75](249.38,10.62)(274.38,10.62)
\rput(256.88,17.5){\program}
\psline[linewidth=0.75](249.38,35.62)(274.38,10.62)
\psline[linewidth=0.75](347.12,113.12)(400.25,113.12)
\psline[linewidth=0.75](400.25,88.12)(400.25,113.12)
\psline[linewidth=0.75](375.25,125)(375.25,113.12)
\rput(375.25,131.25){$\nameslang$}
\rput(390.25,1.25){$\inputt$}
\psline[linewidth=0.75](372.12,88.12)(400.25,88.12)
\psline[linewidth=0.75](390.25,88.12)(390.25,6.88)
\psline[linewidth=0.75](360.25,99.38)(360.25,66.87)
\psline[linewidth=0.75](347.12,113.12)(372.12,88.12)
\psline[linewidth=0.75](330.25,75)(351.5,75)
\psline[linewidth=0.75](371.5,55)(350.25,55)
\psline[linewidth=0.75](351.5,75)(371.5,55)
\psline[linewidth=0.75](330.25,75)(350.25,55)
\psline[linewidth=0.75](350.88,33.12)(350.88,22.5)
\newrgbcolor{userFillColour}{0.2 0 0.2}
\rput{90}(350.88,33.13){\psellipse[fillcolor=userFillColour,fillstyle=solid](0,0)(1.41,-1.4)}
\psline[linewidth=0.75](361.5,55)(361.5,43.75)
\psline[linewidth=0.75](340.88,63.75)(340.88,42.5)
\psline[linewidth=0.75](361.5,43.75)(350.88,33.12)
\psline[linewidth=0.75](340.88,42.5)(350.88,32.5)
\psline[linewidth=0.75](340.25,33.75)(340.25,8.75)
\psline[linewidth=0.75](340.25,8.75)(365.25,8.75)
\rput(347.75,15.62){\program}
\psline[linewidth=0.75](340.25,33.75)(365.25,8.75)
\rput(382.12,100.62){\universal}
\rput(312.5,95){\EQLS}
\rput(66.88,65){$\parteval$}
\rput(350.88,65){$\parteval$}
\newrgbcolor{userLineColour}{1 0.2 0.2}
\psline[linewidth=0.75,linecolor=userLineColour](44.5,104)(44.5,5)
\newrgbcolor{userLineColour}{1 0.2 0.2}
\psline[linewidth=0.75,linecolor=userLineColour](44.5,104)(92.5,56)
\rput(86.5,35){\progH}
\newrgbcolor{userLineColour}{1 0.2 0.2}
\psline[linewidth=0.75,linecolor=userLineColour](92.5,56)(92.5,5)
\newrgbcolor{userLineColour}{1 0.2 0.2}
\psline[linewidth=0.75,linecolor=userLineColour](44.5,5)(92.5,5)
\newrgbcolor{userLineColour}{1 0.2 0.2}
\psline[linewidth=0.75,linecolor=userLineColour](327.5,104)(327.5,5)
\newrgbcolor{userLineColour}{1 0.2 0.2}
\psline[linewidth=0.75,linecolor=userLineColour](327.5,104)(375.5,56)
\rput(369.5,35){\progH}
\newrgbcolor{userLineColour}{1 0.2 0.2}
\psline[linewidth=0.75,linecolor=userLineColour](375.5,56)(375.5,5)
\newrgbcolor{userLineColour}{1 0.2 0.2}
\psline[linewidth=0.75,linecolor=userLineColour](327.5,5)(375.5,5)
\rput(32.5,150){\EQLS}
\rput(132.5,150){\EQLS}
\rput(222.5,150){\EQLS}
\rput(312.5,150){\EQLS}
\rput(0,150){$\oone$}
\rput(82.5,150){$\ttwo$}
\rput(177.5,150){$\tthree$}
\rput(267.5,150){$\ffour$}
\rput(345,150){$\ffive$}
\end{pspicture}

%% file: 4-moncal-coalg.tex

So far, we formalized the \emph{programs $\to\!\!\!\!\to$ computations}\/ correspondence from the left hand column of Table~\ref{table}. But presenting computations in the form $XA\tto{\{F\}}B$ only displays their interfaces, and hides the actual process of computation. To capture that, we switch to the right hand column of Table~\ref{table}, and study the correspondence \emph{adaptive programs $\to\!\!\!\!\to$ processes}. 

A \emph{process}\/ is presented as a morphism in the form $X\otimes A\to X\otimes B$. We interpreted the morphisms in the form $X\otimes A\to B$ as $X$-indexed families of computations with the inputs from $A$ and the outputs in $B$. The indices of type $X$ can be thought of as the states of the world, determining which of the family of computations should be run. Interpreted along the same lines, a process $X\otimes A\tto p X\otimes B$ does not only provide the output of type $B$, but it also updates the state in $X$. This is what state machines also do, and that is why the morphisms $X\times A \tto m X\times B$ in cartesian categories are interpreted as machines. In a sufficiently complete cartesian category, every such machine $m$ induces a 
machine homomorphism $X\tto{\ana m} \xxpp A B$, which assigns to each state $x\in X$ a  \emph{behavior} $\ana m x \in \xxpp A B$, unfolded by the final $AB$-machine $\xxpp A B \times A\tto \xi \xxpp A B \times B$. This was displayed in Table~\ref{table}.  A monoidal computer, though, turns out to provide a much stronger form of representation for its morphisms in the form $X\otimes A\tto p X\otimes B$: each of them induces a 
machine homomorphism $X \tto P \PP$. This $P$ is a \emph{program}\/ for the process $p$. Note that there may be many programs for each process; but on the other hand, all programs, for all processes of all possible input types $A$ and output types $B$, are represented in the same type of programs $\PP$. This makes a fundamental difference, distinguishing machines $m$ from computational processes $p$, which include life \cite{NeumannJ:self,TuringA:morphogenesis}
\footnote{
	Both Turing and von Neumann devoted a lot of attention to studying life as a manifestation of computational processes. Their ideas have been adopted in biology \cite{ArbibM:book-brains,Synth-bio}, but most computer scientists remain skeptical.}
Every family of machines is designed in a suitable engineering language; but all computational processes can be programmed in any Turing complete language, just like all processes of life are programmed in the language of genes. That is why the morphisms $X\otimes A\tto p X\otimes B$ are \emph{processes}, and not merely machines. Their representations $X\tto P \PP$ are not merely $X$-indexed programs, but they are \emph{adaptive}\/ programs, since they adapt to the state changes, in the sense that we now describe.

\begin{definition}\label{def:process}
A morphism $XA\tto{p} XB$ in a monoidal category $\CCc$ is an \emph{$AB$-process}. If $YA\tto{r} YB$ is another $AB$-process, then an $AB$-process homomorphism is a $\CCc$-morphism $X\tto f Y$ such that $(f\otimes A)\comp r = p\comp (f\otimes B)$. We denote by $\CCc_{AB}$ the category of $AB$-processes.
\end{definition}

\begin{definition}\label{def:um}
A \emph{universal process} in a monoidal category $\CCc$ is carried by a \emph{universal state space}\/ $\RR\in |\CCc|$, such that for every pair $A,B\in |\CCc|$ there is a weakly final $AB$-process $\RR A \tto{\RRun^{AB}} \RR B$. The weak finality means that for every process $p\in \CCc(XA,XB)$ there is an \emph{$X$-adaptive program} $P\in \CCc^\natural (X,\RR)$ such that
\bear
\Run{P(x)}_\RR a\ \ & = &\ \  P(p_X(x,a))\\  
\Run{P(x)}_B a\ \  & = & \ \ p_B (x,a)\\[2ex] 
\def\JPicScale{.65}\newcommand{\aah}{\scriptstyle P}\newcommand{\RRsmall}{\scriptstyle \RR}\renewcommand{\dh}{\scriptstyle B}\newcommand{\ahh}{\RRun
 }\newcommand{\bhh}{\scriptstyle A} \newcommand{\xxh}{\scriptstyle X}\input{PIC/uni-mach-p.tex}\ \  & = &\ \ \   \def\JPicScale{.65}\newcommand{\aah}{\scriptstyle P}
 \newcommand{\RRsmall}{\scriptstyle \RR}\newcommand{\xxh}{\scriptstyle X}\renewcommand{\dh}{\scriptstyle B}\newcommand{\ahh}{p}\newcommand{\bhh}{\scriptstyle A} \input{PIC/uni-mach-t.tex} 
\qquad \qquad \qquad 
\begin{tikzar}{row sep=3em,column sep=1.5em}
\& \RR \otimes B\\   
\RR\otimes A \ar{ur}{\RRun}   \& \&  X\otimes B \ar[dashed]{ul}[swap]{P \otimes B} 
 \\ 
 \& X \otimes A \ar{ur}[swap]{p} \ar[dashed]{ul}{P \otimes A} \end{tikzar}
\eear
\end{definition}

\begin{theorem}\label{thm:moncom-coalg}
Let $\CCc$ be a symmetric monoidal category with data services. Then $\CCc$ is a monoidal computer if and only if it has a universal process. The type $\PP$ of programs coincides with the universal state space $\RR$. 
\end{theorem}
%
%
%

\begin{proof}
Given a weakly final $AB$-process $\RR\otimes A \tto{\RRun} \RR\otimes B$, we show that 
\bear
\UK^{AB} & = &  \left(\RR\otimes A \tto{\RRun^{AB}} \RR\otimes B \tto{\cun\otimes B} B\right)
\eear 
is a universal evaluator, and thus makes $\CCc$ into a monoidal computer. Towards proving \eqref{eq:yon}, suppose that we are given a computation $X\otimes A \tto h B$, and consider the process 
\bear
\widehat h & = & \left( X\otimes A \tto {\cmn \otimes A} X\otimes X \otimes A \tto{X\otimes h} X\otimes B\right)
\eear
By Def.~\ref{def:um}, there is then an $X$-adaptive program $H\in \tot \CCc(\RR)$ satisfying  the rightmost equation in the next diagram.
\[
\newcommand{\program}{$\scriptstyle H$}
\newcommand{\gee}{$h$}
\newcommand{\geenohat}{\scriptstyle \widehat h}
\newcommand{\Progtype}{}
\newcommand{\Progk}{\scriptstyle X}
\newcommand{\Progm}{\scriptstyle A}
\newcommand{\Progn}{\scriptstyle B}
\newcommand{\RRsmall}{\scriptstyle \RR}
\newcommand{\ukk}{{\color{red} \scriptstyle\UK}}
\newcommand{\universalmkn}{\scriptstyle \RRun}
\newcommand{\EQLS}{=}
\def\JPicScale{.3}
\input{PIC/computation-H-hat.tex}
\]
The middle equation holds because $H$ is in $\tot \CCc$, i.e. a comonoid homomorphism. Deleting the state update from the process yields \eqref{eq:yon}.

The other way around, if $\CCc$ is a monoidal computer, with universal evaluators for all pairs of types, we claim that the weakly final $AB$-process is  
\bear
\RRun^{AB} & = & \left(\PP\otimes A \tto{\UK^{A(\PP B)}} \PP \otimes B \right)
\eear
To prove the claim, take an arbitrary $AB$-process $X\otimes A \tto p X\otimes B$, and postcompose it with the partial evaluator on $X$, to get
\bear
\widehat p & = & \left(\PP\otimes X\otimes A \tto{\PP \otimes p} \PP \otimes X\otimes B \tto{\SK^{XB\PP}\otimes B}\PP\otimes B \right)
\eear
Using the Fundamental Theorem of Computability, Thm.~\ref{thm:kleene}, construct a Kleene's fixed point $\widehat P\in \CCc(\PP)$ of $\widehat p$.
\newcommand{\inputt}{\scriptstyle A}
\newcommand{\XH}{\scriptstyle X}
\newcommand{\PPh}{\scriptstyle \PP}
\newcommand{\SKh}{\SK}
\newcommand{\universal}{$\UK$}
\newcommand{\gee}{p}
\newcommand{\mhh}{\widehat p}
\newcommand{\nameslang}{\scriptstyle B}
\newcommand{\program}{$\scriptstyle \widehat P$}
\newcommand{\EQLS}{\large =}
\def\JPicScale{.3}
\begin{center}
\input{PIC/kleene-coalg.tex}
\end{center}
The $X$-adaptive program $P\in \tot\CCc(\PP)$ corresponding to the process $p \in \CCc(XA, XB)$ is now $P(x) = \left[\widehat {P}, x\right]^{XB\PP}$.
\newcommand{\Mhh}{\color{red} \scriptstyle P}
\def\JPicScale{.3}
\begin{center}
\input{PIC/kleene-coalg-prog.tex}
\end{center}
This completes the proof that $\UK^{A(B\PP)}$ satisfies definition \ref{def:um} of weakly final $AB$-process, and that $\PP$ is thus not only a type of programs, but also a universal state space.
\end{proof}

%% file: PIC/uni-mach-p.tex
\ifx\JPicScale\undefined\def\JPicScale{1}\fi
\psset{unit=\JPicScale mm}
\psset{linewidth=0.3,dotsep=1,hatchwidth=0.3,hatchsep=1.5,shadowsize=1,dimen=middle}
\psset{dotsize=0.7 2.5,dotscale=1 1,fillcolor=black}
\psset{arrowsize=1 2,arrowlength=1,arrowinset=0.25,tbarsize=0.7 5,bracketlength=0.15,rbracketlength=0.15}
\begin{pspicture}(0,0)(23.12,28.75)
\rput(18.75,28.75){$\dh$}
\psline(6.88,-9.38)(6.88,1.87)
\psline(18.88,6.87)(18.75,25.62)
\rput(16.88,1.87){$\ahh$}
\rput(5,-11.25){$\aah$}
\psline(18.75,-17.5)(18.88,-3.13)
\pspolygon(1.88,-4.39)
(11.86,-14.37)
(1.88,-14.37)(1.88,-4.39)
\psline(1.88,6.87)
(11.88,-3.12)
(23.12,-3.13)
(23.12,6.87)
(5.62,6.87)
(23.12,6.87)
(5.62,6.87)(1.88,6.87)
\psline(6.88,6.88)(6.88,25.62)
\psline(6.88,-18.12)(6.88,-14.38)
\rput(6.88,-21.25){$\xxh$}
\rput[r](6.25,-1.88){$\RRsmall$}
\rput(18.75,-21.25){$\bhh$}
\rput(6.88,28.75){$\RRsmall$}
\end{pspicture}

%% file: PIC/uni-mach-t.tex
\ifx\JPicScale\undefined\def\JPicScale{1}\fi
\psset{unit=\JPicScale mm}
\psset{linewidth=0.3,dotsep=1,hatchwidth=0.3,hatchsep=1.5,shadowsize=1,dimen=middle}
\psset{dotsize=0.7 2.5,dotscale=1 1,fillcolor=black}
\psset{arrowsize=1 2,arrowlength=1,arrowinset=0.25,tbarsize=0.7 5,bracketlength=0.15,rbracketlength=0.15}
\begin{pspicture}(0,0)(20,29.38)
\rput(15,29.38){$\dh$}
\psline(15.62,-17.5)(15.62,-3.12)
\psline(15,6.88)(15,26.25)
\rput(8.75,1.88){$\ahh$}
\psline(-1.88,3.12)
(-1.88,-3.12)
(20,-3.12)
(20,6.88)
(-1.88,6.88)
(16.25,6.86)
(-1.88,6.88)(-1.88,3.12)
\psline(3.12,-18.12)(3.12,-3.12)
\psline(3.12,6.88)(3.12,15.62)
\rput(3.12,-21.25){$\xxh$}
\psline(3.12,21.25)(3.12,26.25)
\rput(1.5,18.88){$\aah$}
\pspolygon(-1.62,25.74)
(8.36,15.75)
(-1.62,15.75)(-1.62,25.74)
\rput[r](1.88,11.25){$\xxh$}
\rput(15.62,-21.25){$\bhh$}
\rput(3.12,29.38){$\RRsmall$}
\end{pspicture}

%% file: PIC/computation-H-hat.tex
\ifx\JPicScale\undefined\def\JPicScale{1}\fi
\psset{unit=\JPicScale mm}
\psset{linewidth=0.3,dotsep=1,hatchwidth=0.3,hatchsep=1.5,shadowsize=1,dimen=middle}
\psset{dotsize=0.7 2.5,dotscale=1 1,fillcolor=black}
\psset{arrowsize=1 2,arrowlength=1,arrowinset=0.25,tbarsize=0.7 5,bracketlength=0.15,rbracketlength=0.15}
\begin{pspicture}(0,0)(360,99.06)
\psline[linewidth=0.75](285,35)(285,0)
\psline[linewidth=0.75](275,35)(255,55)
\psline[linewidth=0.75](255,55)(295,55)
\psline[linewidth=0.75](275,35)(295,35)
\psline[linewidth=0.75](295,55)(295,35)
\psline[linewidth=0.75](285,90)(285,55)
\psline[linewidth=0.75](265,45)(265,15)
\rput(265,-5){$\Progk$}
\rput(285,-5){$\Progm$}
\rput(285,95){$\Progn$}
\rput(280,45){$\universalmkn$}
\psline[linewidth=0.75](255,5)(275,5)
\rput(260,10){\program}
\psline[linewidth=0.75](255,25)(275,5)
\psline[linewidth=0.75](255,25)(255,5)
\psline[linewidth=0.75](10,35)(10,5)
\psline[linewidth=0.75](30,35)(30,5)
\psline[linewidth=0.75](0,35)(0,55)
\psline[linewidth=0.75](0,55)(39.38,55)
\psline[linewidth=0.75](0,35)(39.38,35)
\psline[linewidth=0.75](39.38,55)(39.38,35)
\psline[linewidth=0.75](30,90)(30,55)
\rput(10,0){$\Progk$}
\rput(30,0){$\Progm$}
\rput(30,95){$\Progn$}
\rput(55,44.38){\EQLS}
\rput(20,45){\gee}
\rput(360,-5){}
\psline[linewidth=0.75](95,35)(95,5)
\psline[linewidth=0.75](115,35)(115,5)
\psline[linewidth=0.75](85,35)(85,55)
\psline[linewidth=0.75](85.62,55)(125,55)
\psline[linewidth=0.75](85,35)(124.38,35)
\psline[linewidth=0.75](125,55)(125,35)
\psline[linewidth=0.75](115,90)(115,55)
\rput(95,0){$\Progk$}
\rput(115,0){$\Progm$}
\rput(115,95){$\Progn$}
\rput(105,45){\gee}
\psline[linewidth=0.75](75,50)(75,40)
\psline[linewidth=0.75](95,70)(75,50)
\psline[linewidth=0.75](95,20)(75,40)
\psline[linewidth=0.75](95,80)(95,70)
\rput{90}(95,20){\psellipse[fillstyle=solid](0,0)(1.57,-1.56)}
\rput{90}(95,80){\psellipse[fillstyle=solid](0,0)(1.57,-1.56)}
\psline(70,65)(130,65)
\psline(70,15)(130,15)
\psline(130,15)(130,65)
\psline(70,15)(70,65)
\rput(125,23.75){$\geenohat$}
\rput(235,45){\EQLS}
\psline[linewidth=0.75](265,5)(265,0)
\psline[linewidth=0.75](265,70)(265,55)
\rput{90}(265,70){\psellipse[fillstyle=solid](0,0)(1.57,-1.56)}
\rput[r](92.5,75){$\Progk$}
\newrgbcolor{userLineColour}{1 0 0.4}
\psline[linecolor=userLineColour](250,75)(300,75)
\newrgbcolor{userLineColour}{1 0 0.4}
\psline[linecolor=userLineColour](270,30)(300,30)
\newrgbcolor{userLineColour}{1 0 0.4}
\psline[linecolor=userLineColour](300,30)(300,75)
\newrgbcolor{userLineColour}{1 0 0.4}
\psline[linecolor=userLineColour](250,50)(250,75)
\rput(295,65){$\ukk$}
\rput[r](262.5,65){$\RRsmall$}
\rput[l](267.5,22.5){$\RRsmall$}
\newrgbcolor{userLineColour}{1 0 0.4}
\psline[linecolor=userLineColour](250,50)(270,30)
\rput(145,44.38){\EQLS}
\psline[linewidth=0.75](185,35)(185,5)
\psline[linewidth=0.75](205,35)(205,5)
\psline[linewidth=0.75](175,35)(175,55)
\psline[linewidth=0.75](175.62,55)(215,55)
\psline[linewidth=0.75](175,35)(214.38,35)
\psline[linewidth=0.75](215,55)(215,35)
\psline[linewidth=0.75](205,90)(205,55)
\rput(185,0){$\Progk$}
\rput(205,0){$\Progm$}
\rput(195,45){\gee}
\psline[linewidth=0.75](165,50)(165,40)
\psline[linewidth=0.75](185,70)(165,50)
\psline[linewidth=0.75](185,20)(165,40)
\rput{90}(185,20){\psellipse[fillstyle=solid](0,0)(1.57,-1.56)}
\rput{90}(185,97.5){\psellipse[fillstyle=solid](0,0)(1.57,-1.56)}
\psline(160,65)(220,65)
\psline(160,15)(220,15)
\psline(220,15)(220,65)
\psline(160,15)(160,65)
\rput(215,23.75){$\geenohat$}
\psline[linewidth=0.75](185,97.5)(185,85)
\psline[linewidth=0.75](175,75)(195,75)
\rput(180,80){\program}
\psline[linewidth=0.75](175,95)(195,75)
\psline[linewidth=0.75](175,95)(175,75)
\psline[linewidth=0.75](185,75)(185,70)
\rput[l](187.5,90){$\RRsmall$}
\rput(205,95){$\Progn$}
\end{pspicture}

%% file: PIC/kleene-coalg.tex
\ifx\JPicScale\undefined\def\JPicScale{1}\fi
\psset{unit=\JPicScale mm}
\psset{linewidth=0.3,dotsep=1,hatchwidth=0.3,hatchsep=1.5,shadowsize=1,dimen=middle}
\psset{dotsize=0.7 2.5,dotscale=1 1,fillcolor=black}
\psset{arrowsize=1 2,arrowlength=1,arrowinset=0.25,tbarsize=0.7 5,bracketlength=0.15,rbracketlength=0.15}
\begin{pspicture}(0,0)(210,135.62)
\psline[linewidth=0.75](135,80)(210,80)
\psline[linewidth=0.75](135,45.62)(135,20.62)
\psline[linewidth=0.75](210,55)(210,80)
\psline[linewidth=0.75](170,130)(170,80)
\rput(200,135){$\nameslang$}
\rput(200,7.5){$\inputt$}
\psline[linewidth=0.75](160,55)(210,55)
\psline[linewidth=0.75](170,55)(170,14.38)
\psline[linewidth=0.75](135,20.62)(160,20.62)
\rput(185,67.5){\universal}
\rput(142.5,27.5){\program}
\rput(117.5,65){\EQLS}
\psline[linewidth=0.75](135,45.62)(160,20.62)
\psline[linewidth=0.75](147.5,67.5)(147.5,33.75)
\psline[linewidth=0.75](135,80)(160,55)
\psline[linewidth=0.75](40,80)(90,80)
\psline[linewidth=0.75](90,55)(90,80)
\psline[linewidth=0.75](80,130.62)(80,80)
\rput(80,135.62){$\nameslang$}
\rput(80,8.12){$\inputt$}
\psline[linewidth=0.75](40,55)(90,55)
\psline[linewidth=0.75](80,55)(80,14.38)
\psline[linewidth=0.75](20,100)(20,33.75)
\psline[linewidth=0.75](40,80)(40,55)
\rput(65,66.88){$\gee$}
\psline[linewidth=0.75](18.12,120)(43.12,120)
\psline[linewidth=0.75](37.5,100)(62.5,100)
\psline[linewidth=0.75](18.12,120)(38.12,100)
\psline[linewidth=0.75](43.12,120)(63.12,100)
\psline[linewidth=0.75](51.88,99.38)(51.88,80)
\psline[linewidth=0.75](20,100)(29.38,108.75)
\psline[linewidth=0.75](51.88,130.63)(51.88,111.25)
\psline[linewidth=0.75](51.88,54.38)(51.88,15)
\psline[linewidth=0.75](8.75,45.62)(8.75,20.62)
\psline[linewidth=0.75](8.75,20.62)(33.75,20.62)
\rput(16.25,27.5){\program}
\psline[linewidth=0.75](8.75,45.62)(33.75,20.62)
\psline[linewidth=0.5](15,125)(95,125)
\psline[linewidth=0.5](15,50)(95,50)
\psline[linewidth=0.5](95,125)(95,50)
\psline[linewidth=0.5](15,125)(15,50)
\psline[linewidth=0.75](200,55)(200,14.38)
\psline[linewidth=0.75](200,130)(200,80)
\rput(170,135){$\PPh$}
\rput(51.88,135){$\PPh$}
\rput(51.88,8.75){$\XH$}
\rput(170,8.12){$\XH$}
\rput[r](50,90){$\XH$}
\rput(88.12,110){$\mhh$}
\rput(40,110){$\SKh$}
\end{pspicture}

%% file: PIC/kleene-coalg-prog.tex
\ifx\JPicScale\undefined\def\JPicScale{1}\fi
\psset{unit=\JPicScale mm}
\psset{linewidth=0.3,dotsep=1,hatchwidth=0.3,hatchsep=1.5,shadowsize=1,dimen=middle}
\psset{dotsize=0.7 2.5,dotscale=1 1,fillcolor=black}
\psset{arrowsize=1 2,arrowlength=1,arrowinset=0.25,tbarsize=0.7 5,bracketlength=0.15,rbracketlength=0.15}
\begin{pspicture}(0,0)(210,135.62)
\psline[linewidth=0.75](157.5,80)(210,80)
\psline[linewidth=0.75](210,55)(210,80)
\psline[linewidth=0.75](170,130)(170,80)
\rput(200,135){$\nameslang$}
\rput(200,7.5){$\inputt$}
\psline[linewidth=0.75](183.12,55)(210,55)
\rput(190.62,66.88){\universal}
\rput(110,65){\EQLS}
\psline[linewidth=0.75](157.5,80)(182.5,55)
\psline[linewidth=0.75](40,80)(90,80)
\psline[linewidth=0.75](90,55)(90,80)
\psline[linewidth=0.75](80,130.62)(80,80)
\rput(80,135.62){$\nameslang$}
\rput(80,8.12){$\inputt$}
\psline[linewidth=0.75](40,55)(90,55)
\psline[linewidth=0.75](80,55)(80,14.38)
\psline[linewidth=0.75](40,80)(40,55)
\rput(65,66.88){$\gee$}
\psline[linewidth=0.75](18.12,120)(43.12,120)
\psline[linewidth=0.75](37.5,100)(62.5,100)
\psline[linewidth=0.75](18.12,120)(38.12,100)
\psline[linewidth=0.75](43.12,120)(63.12,100)
\psline[linewidth=0.75](51.88,99.38)(51.88,80)
\psline[linewidth=0.75](27.5,103.75)(27.5,110.62)
\psline[linewidth=0.75](51.88,130.63)(51.88,111.25)
\psline[linewidth=0.75](51.88,54.38)(51.88,15)
\psline[linewidth=0.75](16.25,115.62)(16.25,90.62)
\psline[linewidth=0.75](16.25,90.62)(41.25,90.62)
\rput(25,96.88){\program}
\psline[linewidth=0.75](16.25,115.62)(41.25,90.62)
\newrgbcolor{userLineColour}{1 0.2 0.2}
\psline[linewidth=0.5,linecolor=userLineColour](11.25,133.75)(35,133.75)
\newrgbcolor{userLineColour}{1 0.2 0.2}
\psline[linewidth=0.5,linecolor=userLineColour](11.25,85.62)(67.5,85.62)
\newrgbcolor{userLineColour}{1 0.2 0.2}
\psline[linewidth=0.5,linecolor=userLineColour](11.25,133.75)(11.25,85.62)
\psline[linewidth=0.75](200,55)(200,14.38)
\psline[linewidth=0.75](200,130)(200,80)
\rput(170,135){$\PPh$}
\rput(51.88,135){$\PPh$}
\rput(51.88,8.75){$\XH$}
\rput(170,7.5){$\XH$}
\rput(40,110){$\SKh$}
\newrgbcolor{userLineColour}{1 0.2 0.2}
\psline[linewidth=0.5,linecolor=userLineColour](67.5,85.62)(67.5,100.62)
\newrgbcolor{userLineColour}{1 0.2 0.2}
\psline[linewidth=0.5,linecolor=userLineColour](35,133.75)(67.5,101.25)
\rput[tl](13.12,131.88){$\Mhh$}
\psline[linewidth=0.75](136.25,56.25)(161.25,56.25)
\psline[linewidth=0.75](155.62,36.25)(180.62,36.25)
\psline[linewidth=0.75](136.25,56.25)(156.25,36.25)
\psline[linewidth=0.75](161.25,56.25)(181.25,36.25)
\psline[linewidth=0.75](170,35.62)(170,14.38)
\psline[linewidth=0.75](145.62,40)(145.62,46.88)
\psline[linewidth=0.75](170,66.88)(170,47.5)
\psline[linewidth=0.75](134.38,51.88)(134.38,26.88)
\psline[linewidth=0.75](134.38,26.88)(159.38,26.88)
\rput(143.12,33.12){\program}
\psline[linewidth=0.75](134.38,51.87)(159.38,26.87)
\newrgbcolor{userLineColour}{1 0.2 0.2}
\psline[linewidth=0.5,linecolor=userLineColour](129.38,70)(153.12,70)
\newrgbcolor{userLineColour}{1 0.2 0.2}
\psline[linewidth=0.5,linecolor=userLineColour](129.38,21.88)(185.62,21.88)
\newrgbcolor{userLineColour}{1 0.2 0.2}
\psline[linewidth=0.5,linecolor=userLineColour](129.38,70)(129.38,21.88)
\rput(158.12,46.25){$\SKh$}
\newrgbcolor{userLineColour}{1 0.2 0.2}
\psline[linewidth=0.5,linecolor=userLineColour](185.62,21.88)(185.62,36.88)
\newrgbcolor{userLineColour}{1 0.2 0.2}
\psline[linewidth=0.5,linecolor=userLineColour](153.12,70)(185.62,37.5)
\rput[tl](131.25,68.12){$\Mhh$}
\end{pspicture}

%% file: 5-moncal-turing.tex
In the remaining two sections we show how to run Turing machines in a monoidal computer, and how to measure their complexity. But a coalgebraic treatment of Turing machines as machines, in the sense discussed at the beginning of Sec.~\ref{Sec-four}, would only display their behaviors, i.e. what rewrite and which move of the machine head will happen on which input, and it obliterates the configurations of the tape, where the actual computation happens. In terms of Sec.~\ref{Sec-four}, a Turing machine as a model of actual computation should not be viewed as a machine, but as a process. So we call them \emph{Turing processes}\/ here. While changing well established terminology is seldom a good idea, and we may very well regret this decision, the hope is that it will be a useful reminder that we are doing something unusual: relating Turing machines with  adaptive programs, coalgebraically. The presented constructions go through in an arbitrary monoidal computer, but require spelling out a suitable representation of the integers, and some arithmetic. This was done in \cite {PavlovicD:IC12}, and can be done more directly; but for the sake of brevity, we work here with the category $\CC$ of recursively enumerable sets and computable partial functions from Sec.~\ref{sec:examples}. 
The monoidal structure and the data services are induced by the cartesian products of sets, which are, of course not categorical products in a category of partial functions\footnote{The reason is that the singleton set, which is still the tensor unit, is not a terminal object for partial functions.}. 
The monoidal category $(\CCc, \otimes, I)$ will henceforth thus be $(\CC,\otimes, \OOne)$.

Recall that Turing's definition of his machines can be recast \cite[Appendix]{PavlovicD:AMAST06}
to processes in the form
\[
Q_\tu \otimes \Sigma \ppfn \tu Q_\tu \otimes \Sigma \otimes \Theta
\]
where 
\begin{itemize}
\item $Q_\tu$ is the finite set of states, always including the \emph{final}\/ state $\checkmark \in Q_\tu$;
\item $\Sigma$ is a fixed alphabet, the same for all $\tu$, 
always including the blank symbol $\blank\in \Sigma$;
\item $\Theta = \{\tleft,\tstall,\tright\}$ are the directions in which the head can move along the tape.
\end{itemize}
%
Let us recall the execution model: how these machines and processes compute. A Mealy machine $Q_\kappa \times I \ppfn \kappa Q_\kappa \times O$ inputs a string $n \tto{\iota} I$, where $n = \{0,1,\ldots,n-1\}$ sequentially, e.g. it reads the inputs $\iota_0$, then $\iota_1$ etc, and it outputs a string $n\tto {\omega} O$ in the same order, i.e. $\omega_0$, $\omega_1$, etc.  In contrast, a Turing process in principle overwrites its inputs, and outputs the results of overwriting when it halts; therefore, in a Turing process, the input alphabet $I$ and its output alphabet $O$ must be the same, say $I = O = \Sigma$. Both the inputs, and the outputs, and the intermediary data of a Turing process are in the form $w: \ZZz\to \Sigma$, where all but finitely many values $w(z)$ must be $\blank$. So each word $w: \ZZz\to \Sigma$ is still a finite string of symbols, like in the Mealy machine model. The difference is that $w$ is written on the infinite \emph{'tape'}, here represented by the set of integers $\ZZz$,  which allows the processing \emph{'head'}\/ to move in both directions, or to stay stationary (while in a Mealy machine the head moves in the same direction at each step). We represent the position of the head by the integer $0$, and the symbol that the head reads on that position is thus denoted by $w(0)$. If the process $Q_\tu \otimes \Sigma \ppfn\tu  Q_\tu \otimes \Sigma \otimes \Theta$, which is a triple of  functions $\tu = <\tu_Q, \tu_\Sigma, \tu_\Theta>$, is defined on a given state $q\in Q_\tu$ and a given input  $\sigma = w(0)$, then it will
\begin{itemize}
\item overwrite $\sigma$ with $\sigma ' = \tu_\Sigma(q,\sigma)$,
\item transition to the state $q' = \tu_Q(q,\sigma)$, and
\item move the head to the next cell in the direction $\theta = \tu_\Theta(q,\sigma)$.
\end{itemize}
If $q = \checkmark$, then $\tu(\checkmark, \sigma) = <\checkmark, \sigma, \tstall>$, which means that the process must halt at the state $\checkmark$, if it ever reaches it.

To capture this execution model formally, we extend Turing processes over the alphabet $\Sigma$, first to processes over the set $\WW$ of $\Sigma$-words written on a tape, and then to computations with the inputs and the outputs from $\WW$
\[
\prooftree
\prooftree
\hspace{1.5em}Q_\tu \otimes \Sigma \ppfn\tu  Q_\tu \otimes \Sigma \otimes \Theta
\justifies
Q_\tu \otimes \WW \ppfn{\tuu}  Q_\tu \otimes \WW
\endprooftree
\justifies
Q_\tu \otimes \WW \ppfn{\tuuu} \WW\hspace{2.3em}
\endprooftree
\]
where 
\bear 
\WW & = & \Big\{ w:\ZZz \to \Sigma \ |\ \supp(w)\lt \infty\Big\} 
\eear
is the set of $\Sigma$-words written on a tape, and $\supp(w) = \{z\ |\ w(z) \neq \blank\}$ . The elements of $\WW$ are often also called the \emph{tape configurations}. Writing the tuples in the form $\tuu = <\tuu_Q, \tuu_{\WW}>$, define
\bear
\tuu_Q (q,w) & = & \tu_Q(q, w(0))\\
\tuu_{\WW}(q,w) & = & w' \mbox{ where }
w'(z) = \left.\begin{cases}
\widetilde w(z-1) & \mbox{ if }\tu_\Theta\left(q, w(0)\right) = \tleft\\
\widetilde w(z) & \mbox{ if }\tu_\Theta\left(q, w(0)\right) = \tstall\\
\widetilde w(z+1) & \mbox{ if }\tu_\Theta\left(q, w(0)\right) = \tright\\
\end{cases} \right\} \mbox{ and }\\
&&\hspace{4.5em} \widetilde w(z) = \left.\begin{cases}
\tu_\Sigma\left(q, w(0)\right) & \mbox{ if } z = 0\\
w(z) & \mbox{ otherwise}
\end{cases}\right\}\\[2ex]
\tuuu(q,w)& = & \begin{cases}
w & \mbox{ if } q=\checkmark\\
\tuuu\big(\tuu(q,w)\big) & \mbox{ otherwise}
\end{cases}
\eear
The execution of all Turing processes  can now be captured as a single process
\[
\Turm \otimes \WW \ppfn \tum \Turm\otimes \WW
\]
where the state space $\Turm$ is the disjoint union of the state spaces $Q_\tu$ of all Turing processes $\tu\in \TM$, i.e.
\bear
\Turm  & = & \coprod_{\tu \in \TM} Q_\tu\mbox{ where }
\TM = \{Q_\tu \otimes \Sigma \ppfn{\tu} Q_\tu \otimes \Sigma \otimes \Theta \}
\eear
so that the elements of $\Turm$ are the pairs $<\tu, q>$, where $q\in Q_\tu$, and $\Turm\otimes \WW \ppfn \tum \Turm\otimes \WW$ is the pair $\tum = <\tum_\Turm, \tum_{\WW}>$ which, when applied to $<\tu, q>\in \Turm$ and $w\in \WW$, gives:
\bear
\tum \big(<\tu, q>, w \big) & = & \big<<\tu, q'>, w'\big> \ \mbox{ where } q' = \tuu_Q(q,w) \mbox{ and } w' = \tuu_\WW(q,w)
\eear
By applying Thm.~\ref{thm:moncom-coalg} to the process $\Turm \otimes \WW \ppfn \tum \Turm\otimes \WW$, we get the following 

\begin{proposition}\label{prop:TM}
There is an adaptive program $\Turmm\in \tot \CC(\Turm, \PP)$ such that $\Turmm(\tu, q)$ executes any Turing process $\tu$ starting from the initial state $q\in Q_\tu$. This means that for every tape configuration $w \in \WW$ holds
\bear
\Run{\Turmm(\tu, q)}_\PP\, w \,  & = & \Turmm(\tu, q')\\
\Run{\Turmm(\tu, q)}_\WW w & = & w'
\eear
where $q' = \tu_Q\big(q, w(0)\big)$ is the next state of $\tu$, and $w'=\tuu_\WW\big(q, w \big)$ is the next tape configuration. (The string diagram is the same as the one in Def.~\ref{def:um}.)
\end{proposition}


\begin{corollary}
The monoidal computer $\CC$ is Turing complete.
\end{corollary}

%% file: 6-moncal-complexity.tex

\subsection{Evaluating Turing processes}
Using the process $\Turm \otimes \WW \ppfn \tum \Turm\otimes \WW$, which according to Prop.~\ref{prop:TM} executes the single step transitions of Turing processes, we would now like to define a computation $\Turm \otimes \WW \ppfn{\tummm} \WW$ that will evaluate Turing processes all the way; i.e. should execute all transitions that a process executes, and halt and deliver the output if the process halts, or diverge if the process diverges. The idea is to run something like the following pseudocode
\bea\label{eq:ttt}
\tummm\big(<\tu, q>, w\big)& = &\Big(x:= <\tu, q>;\  y:=w;\\  
&&\ \  {\tt while}\ \big(\tum_\Turm(x,y)\neq \checkmark\big) \notag\\
&&\ \ \ \ \ \ \ \ \ \ \ \Big\{x:= \tum_\Turm(x,y);\ y:= \tum_\WW(x,y) \Big\};\notag\\  
&& \ \ {\tt print}\ y\Big)\notag
\eea
%
%
We implement this program using the Fundamental Theorem of Computability. 
The function $\tummm$ is derived as a Kleene's fixed program for an intermediary function $\tumm$, lifting the derivation from Sec.~\ref{Sec-five}.
\[
\prooftree
\prooftree
\hspace{3.5em} \Turm \otimes \WW \ppfn \tum \Turm\otimes \WW
\justifies
\PP \otimes \Turm \otimes \WW 
\ppfn{\tumm} \WW 
\endprooftree
\justifies
\hspace{2em} \Turm \otimes \WW 
\ppfn{\tummm} \WW 
\endprooftree
\]
The definition of $\tumm$ lifts the definition of $\tuu$ from Sec.~\ref{Sec-five}, extended by an undetermined program $\Upsilon$
\bear
\tumm\big(\Upsilon, <\tu, q>, w \big) & = &  \begin{cases}
w & \mbox{ if } \tu_Q\left(q,w(0)\right) = \checkmark\\
\{\Upsilon\}\big(<\tu, q'>,\  w' \big) & \mbox{ otherwise}\\
\ \ \mbox{ where } q' = \tu_Q\big(q, w(0)\big) \\
\ \ \mbox{ and }\ \  w' = \tum_\WW\big(<\tu, q>, w\big)
\end{cases}
\eear
Using the $\iif$-branching from Sec.~\ref{sec:defmoncom}, this schema can be expressed in a monoidal computer, as illustrated in the diagram below. Set $\Upsilon$ to be Kleene's fixed program $\Turmm$ of $\tumm$, and define $\tummm = \{\Turmm\}$.  This construction boils down to the first one of the following string diagram equations:
\medskip

\[
\newcommand{\iiftag}{\iif = \UK}
\newcommand{\gtag}{\neg \checkmark ?}
\newcommand{\predtag}{\tum}
\newcommand{\thh}{{\color{red} \tumm}}
\newcommand{\tth}{{\color{red} \tummm}}
\newcommand{\Thh}{\scriptstyle \Turmm}
\newcommand{\uuh}{{\color{red}\Turmmm}}
\newcommand{\pee}{\scriptstyle \PP}
\newcommand{\een}{\scriptstyle \Turm}
\newcommand{\eex}{\scriptstyle \WW}
\newcommand{\utag}{ \UK}
\newcommand{\shh}{ \SK}
\newcommand{\EQLS}{=}
\def\JPicScale{.27}
\input{PIC/tHat-WH.tex}
\]
\medskip

\noindent Given $<\tu, q> \in \Turm$ and $w\in \WW$, $\tummm$ thus runs $\tu$ on $w$, starting from $q$ and halting at $\checkmark$, at which point it outputs the current $w$. If it does not reach $\checkmark$, then $\tu$ runs forever. The second equation in the above diagram proves the next proposition.

\begin{proposition}\label{prop:TM-eval}
There is an adaptive program $\Turmmm\in \tot \CC(\Turm, \PP)$ that evaluates any Turing process $\tu$ starting from a given initial state $q\in Q_\tu$. This means that for every tape configuration $w \in \WW$ holds
\bear
\big\{\Turmmm(\tu, q)\big\} w & = & \tuuu( q, w)
\eear
\end{proposition}

\subsection{Counting time}
To count the steps in the executions of Turing processes, we add a counter $i\in \NNn$ to the Turing process evaluator $\tummm$. The counter gets increased by at each execution step, and thus counts them. We call $\timmm$ the computation which outputs the final count. If $\tummm$ halts, then $\timmm$ outputs the value of the counter $i$; if $\tummm$ does not halt, then $\timmm$ diverges as well. The pseudocode for $\timmm$ could thus look something like this: 
\bea\label{eq:tttt}
\timmm\big(<\tu, q>, w\big)& = &\Big(x:= <\tu, q>;\  y:=w;\ i:=0; \\  
&&\ \  {\tt while}\ \big(\tum_\Turm(x,y)\neq \checkmark\big) \notag\\
&&\hspace{3.75em} \Big\{x:= \tum_\Turm(x,y);\ y:= \tum_\WW(x,y);\ i:=i+1 \Big\};\notag\\  
&&\ \  {\tt print}\ i\Big)\notag
\eea
The implementation of $\timmm$ in a monoidal computer is similar to the implementation of $\tummm$. It follows a similar derivation pattern:
\[
\prooftree
\prooftree
\hspace{3.5em} \Turm \otimes \WW 
\ppfn \tum \Turm\otimes \WW 
\justifies
\PP \otimes \Turm \otimes \WW \otimes \NNn  
\ppfn{\timm} \NNn\hspace{2em}
\endprooftree
\justifies
\hspace{2em} \Turm \otimes \WW 
\ppfn{\timmm} \NNn 
\endprooftree
\]
where
\bear
\timm\big(\Upsilon, <\tu, q>, w, i \big) & = &  \begin{cases}
i & \mbox{ if } \tu_Q\left(q,w(0)\right) = \checkmark\\
\{\Upsilon\}\big(<\tu, q'>,\  w', \ i+1 \big) & \mbox{ otherwise}
\end{cases}
\\[2ex]
\timmm\big(<\tu, q>, w\big) &  = &  \big\{\Timm\big\} \big(<\tu, q>, w, 0\big)
\eear
where 
\bear
q' & = & \tu_Q\big(q, w(0)\big) \\
w' & = & \tum_\WW\big(<\tu, q>, w\big)
\eear

and $\Timm$ is Kleene's fixed program of $\timm$. It is easy to see, and prove, that $\timmm\big(<\tu, q>, w\big)$ halts if and only if $\tuuu(q,w)$ halts, and if it does halt, then it outputs the number of steps that $\tu$ made before halting, having started from $q$ and $w$. The string diagrams that implmenet  $\timm$, $\Timm$, $\timmm$ and $\Timmm$ are obtained from the string diagrams in the previous section by renaming $p$s to $t$s and $P$s to $T$s, and by adding string of type $\NNn$ to the right, with one operation on it, increasing the counter. This string, of course, outputs the time complexity $\tummm$. Hence


\begin{proposition}\label{prop:CX-eval}
There is an adaptive program $\Timmm\in \tot \CC(\Turm, \PP)$ that outputs the number of steps that a Turing process $\tu$ makes in any run from a given initial state $q\in Q_\tu$ to the halting state $\checkmark$. If the Turing process $\tu$ starting from $q$ diverges, then the computation $\{\Timmm(\tu,q)\}$ diverges as well. This means that, for every tape configuration $w \in \WW$ holds
\bear
\big\{\Timmm(\tu, q)\big\} w & = & \timmm\big(<\tu, q>, w\big)
\eear
\end{proposition}

\subsection{Counting space}
So far, we used the integers $\ZZz$ as the index set for the tape configurations $w:\ZZz\to \Sigma$. The position of the head has always been $0\in \ZZz$, and whenever the head moves, the tape configuration $w$ gets updated to $w' = \tuu_{\WW}(q,w)$, where $w'(0)$ is the new position of the head, and the rest of the word $w$ is reindexed accordingly, as described in Sec.~\ref{Sec-five}. At each point of the computation $w$ thus describes the tape content \emph{relative to the current position of the head}; there is no record of the prior positions or contents.

To count the tape cells used by Turing processes,  we must make the tape itself into a first class citizen. The simplest way to do this seems to be to add a counter $m\in \ZZz$, which denotes the offset of the current position of the head with respect to the initial position. This allows us to record how far up and down the tape, how far from its original position, does the head ever travel in either direction during  the computation. To record these maximal offsets of the head, we need two more counters: let $r\in \ZZz$ be the highest value that the head offset $m$ ever takes; and let $\ell\in \ZZz$ be the lowest value that the head offset $m$ ever takes. The number of cells that the head has visited during the computation is then clearly $r-\ell$. To implement this space counting idea, we need to run a program roughly like this:
\bea\label{eq:space}
\spaccc\big(<\tu, q>, w\big)& = &\Bigg(x:= <\tu, q>;\  y:=w;\ \ell, m, r:=0; \\  
&&\ \  {\tt while}\ \big(\tum_\Turm(x,y)\neq \checkmark\big) \notag\\
&&\hspace{3.75em} \Bigg\{x:= \tum_\Turm(x,y);\ y:= \tum_\WW(x,y);\notag \\
&&\hspace{3.75em} {\tt if}\ \big(\tu_\Theta\big(q,w(0)\big) = \tleft\big) \notag \\
&&\hspace{4.75em} \Big\{{\tt if}\ (m = \ell)\{ \ell := \ell-1\};\ m:=m-1\Big\} \notag \\
&&\hspace{3.75em} {\tt if}\ \big(\tu_\Theta\big(q,w(0)\big) = \tright\big) \notag \\
&&\hspace{4.75em} \Big\{{\tt if}\ (m = r)\{ r := r+1\};\ m:=m+1\Big\} \Bigg\}\notag \\
&&\ \  {\tt print}\ r-\ell\Bigg)\notag
\eea
The derivation now becomes
\[
\prooftree
\prooftree
\hspace{3.5em} \Turm \otimes \WW 
\ppfn \tum
\Turm\otimes \WW 
\justifies
\PP \otimes \Turm \otimes \WW \otimes \ZZz^3  
\ppfn{\spacc} \NNn\hspace{2em}
\endprooftree
\justifies
\hspace{2em} \Turm \otimes \WW 
\ppfn{\spaccc} \NNn 
\endprooftree
\]
where
\bear
\spacc\big(\Upsilon, <\tu, q>, w, \ell, m, r \big) & = &  \left.\begin{cases}
r-\ell & \mbox{ if } \tu_Q\left(q,w(0)\right) = \checkmark\\
\{\Upsilon\}\big(<\tu, q'>,\  w', \ \ell',\ m',\ r' \big) & \mbox{ otherwise}\end{cases}\right\} 
\\[2ex]
\spaccc\big(<\tu, q>, w\big) &  = &  \big\{\Spacc\big\} \big(<\tu, q>, w, 0,0,0\big)
\eear
where
\bear
q' & = & \tu_Q\big(q,w(0)\big) \\
w' & = & \tum_\WW\big(<\tu, q>, w\big)\\
\ell' & = &\begin{cases}
\ell - 1 & \mbox{ if } m = \ell \mbox{ and } \tu_\Theta\big(q,w(0)\big) = \tleft \\
\ell  & \mbox{ otherwise }
\end{cases} \\[2ex]
m'  & = & \begin{cases}
m - 1 & \mbox{ if } \tu_\Theta\big(q,w(0)\big) = \tleft \\
m & \mbox{ if } \tu_\Theta\big(q,w(0)\big) = \tstall \\
m+1 & \mbox{ if } \tu_\Theta\big(q,w(0)\big) = \tright
\end{cases} \\[2ex]
r' & = & \begin{cases}
r+1 & \mbox{ if } m = r \mbox{ and } \tu_\Theta\big(q,w(0)\big) = \tright \\
r  & \mbox{ otherwise }
\end{cases}
\eear
and $\Spacc$ is Kleene's fixed program of $\spacc$. In a monoidal computer, the above constructions correspond to the following diagrams

\[
\newcommand{\iiftag}{\iif = \UK}
\newcommand{\gtag}{\neg \checkmark ?}
\newcommand{\predtag}{\tum}
\newcommand{\thh}{{\color{red} \spacc}}
\newcommand{\tth}{{\color{red} \spaccc}}
\newcommand{\Thh}{\scriptstyle \Spacc}
\newcommand{\uuh}{{\color{red}\Spaccc}}
\newcommand{\pee}{\scriptstyle \PP}
\newcommand{\een}{\scriptstyle \Turm}
\newcommand{\ees}{\scriptstyle \WW}
\newcommand{\eex}{\scriptstyle \NNn}
\newcommand{\utag}{ \UK}
\newcommand{\shh}{ \SK}
\newcommand{\zzero}{\scriptstyle 0}
\newcommand{\subtract}{-}
\newcommand{\priming}{()'}
\newcommand{\elprime}{\scriptscriptstyle \ell'}
\newcommand{\mprime}{\scriptscriptstyle m'}
\newcommand{\rprime}{\scriptscriptstyle r'}
\newcommand{\elnoprime}{\scriptscriptstyle \ell}
\newcommand{\rldif}{\scriptscriptstyle r-\ell}
\newcommand{\rnoprime}{\scriptscriptstyle r}
\newcommand{\Ints}{\scriptstyle \ZZz^3}
\newcommand{\EQLS}{=}
\def\JPicScale{.28}
\input{PIC/sHat-WH.tex}
\]

\medskip
\noindent The box $()'$, which computes $\ell'$, $m'$ and $r'$ as above, is implemented by composing several branching commands, e.g. as described at the end of Sec.~\ref{sec:defmoncom}. Implementing this box is an easy but instructive exercise in programming monoidal computers. Put together, these constructions prove the following proposition.

\begin{proposition}\label{prop:SX-eval}
There is an adaptive program $\Spaccc\in \tot \CC(\Turm, \PP)$ that outputs the number of cells that a Turing process $\tu$ uses in any run from a given initial state $q\in Q_\tu$ to the halting state $\checkmark$. If the Turing process $\tu$ starting from $q$ diverges, then the computation $\{\Spaccc(\tu,q)\}$ diverges as well. This means that, for every tape configuration $w \in \WW$ holds
\bear
\big\{\Spaccc(\tu, q)\big\} w & = & \spaccc\big(<\tu, q>, w\big)
\eear
\end{proposition}

\begin{remark} There are many variations of the above definitions in the literature, and several different counting conventions. E.g., an alternative to the above definition of $\spaccc$ would be something like
\bear \spaccc'\big(<\tu, q>, w\big)  & = &  \big\{\Spacc\big\} \big(<\tu, q>, w, w_\ell,0,w_r\big)\eear
where
\bear
w_\ell & = & \min\{i\in \ZZz\ |\ w(i)\neq \blank\}\\
w_r & = & \max\{i\in \ZZz\ |\ w(i)\neq \blank\}
\eear
In contrast with $\spaccc$, where the space counting convention is that a memory cell counts as used if and only if it is ever reached by the head, the space counting convention behind $\spaccc'$ is that every computation uses at least $|w| = w_r-w_\ell$ cells, on which its initial input is written. If a Turing process halts without reading all of its input $w$, or even without reading any of it, the space used will still be $|w|$. Some textbooks adhere to the $\spaccc$-counting convention, some to the $\spaccc'$-counting convention, but many do not describe the process in enough detail to discern this difference. This is perhaps justified by the fact that the resulting complexity classes and their hierarchies are the same for all such subtly different counting conventions. E.g., the difference between $\spaccc$ and $\spaccc'$ is absorbed by the $\OOO$-notation, and only arises for computations that do not read their inputs. 
\end{remark}

%% file: PIC/tHat-WH.tex
\ifx\JPicScale\undefined\def\JPicScale{1}\fi
\psset{unit=\JPicScale mm}
\psset{linewidth=0.3,dotsep=1,hatchwidth=0.3,hatchsep=1.5,shadowsize=1,dimen=middle}
\psset{dotsize=0.7 2.5,dotscale=1 1,fillcolor=black}
\psset{arrowsize=1 2,arrowlength=1,arrowinset=0.25,tbarsize=0.7 5,bracketlength=0.15,rbracketlength=0.15}
\begin{pspicture}(0,0)(395,175)
\psline[linewidth=0.75](35,120)(85,120)
\psline[linewidth=0.75](85,100)(85,120)
\psline[linewidth=0.75](55,100)(85,100)
\psline[linewidth=0.75](45,110)(0,65)
\psline[linewidth=0.75](35,120)(55,100)
\rput(70,110){$\utag$}
\psline[linewidth=0.6](-15,120)(15,120)
\psline[linewidth=0.6](15,100)(15,120)
\psline[linewidth=0.6](-15,100)(15,100)
\psline[linewidth=0.6](-15,120)(-15,100)
\psline[linewidth=0.75](80,100)(80,65)
\psline[linewidth=0.6,border=2](0,95)(60,75)
\psline[linewidth=0.75](0,145)(0,120)
\psline[linewidth=0.6](-10,155)(145,155)
\psline[linewidth=0.6](145,135)(145,155)
\psline[linewidth=0.6](-10,155)(10,135)
\psline[linewidth=0.75](105,175)(105,155)
\psline[linewidth=0.6](85,45)(85,65)
\psline[linewidth=0.6](55,45)(85,45)
\psline[linewidth=0.6](55,65)(55,45)
\psline[linewidth=0.75](0,65)(0,10)
\rput(0,110){$\gtag$}
\rput(70,145){$\iiftag$}
\rput(70,55){$\predtag$}
\rput[l](2.5,15){$\pee$}
\rput[r](57.5,0){$\een$}
\rput[l](82.5,0){$\eex$}
\psline[linewidth=0.75](0,100)(0,95)
\psline[linewidth=0.75](60,100)(60,75)
\psline[linewidth=0.6](10,135)(145,135)
\psline[linewidth=0.75](60,45)(60,-5)
\psline[linewidth=0.75](70,135)(70,120)
\psline[linewidth=0.6,border=1.5](80,35)(140,55)
\psline[linewidth=0.75,border=0.8](60,75)(60,65)
\psline[linewidth=0.75,border=0.9](80,-5)(80,45)
\newrgbcolor{userFillColour}{0.2 0 0.2}
\rput{90}(60,75){\psellipse[fillcolor=userFillColour,fillstyle=solid](0,0)(2.58,-2.57)}
\newrgbcolor{userFillColour}{0.2 0 0.2}
\rput{90}(80,35){\psellipse[fillcolor=userFillColour,fillstyle=solid](0,0)(2.58,-2.57)}
\rput[r](67.5,127.5){$\eex$}
\rput[l](82.5,85){$\eex$}
\rput[r](57.5,85){$\een$}
\rput[r](-2.5,127.5){$\pee$}
\rput[r](137.5,127.5){$\eex$}
\rput[l](107.5,170){$\eex$}
\psline[linewidth=0.6](55,65)(85,65)
\psline[linewidth=0.75](140,135)(140,55)
\psline[linewidth=0.75](-10,0)(10,0)
\psline[linewidth=0.75](-10,0)(-10,20)
\psline[linewidth=0.75](-10,20)(10,0)
\newrgbcolor{userLineColour}{1 0.2 0.2}
\psline[linecolor=userLineColour](150,25)(150,160)
\newrgbcolor{userLineColour}{1 0.2 0.2}
\psline[linecolor=userLineColour](-20,160)(150,160)
\newrgbcolor{userLineColour}{1 0.2 0.2}
\psline[linecolor=userLineColour](-20,25)(-20,160)
\newrgbcolor{userLineColour}{1 0.2 0.2}
\psline[linecolor=userLineColour](-20,25)(150,25)
\rput(-5,5){$\Thh$}
\rput(137.5,35){$\thh$}
\rput(170,100){$\EQLS$}
\psline[linewidth=0.75](250,175)(250,120)
\rput[l](252.5,170){$\eex$}
\newrgbcolor{userLineColour}{1 0.2 0.2}
\psline[linecolor=userLineColour](280,65)(280,130)
\newrgbcolor{userLineColour}{1 0.2 0.2}
\psline[linecolor=userLineColour](190,130)(280,130)
\newrgbcolor{userLineColour}{1 0.2 0.2}
\psline[linecolor=userLineColour](190,65)(280,65)
\newrgbcolor{userLineColour}{1 0.2 0.2}
\psline[linecolor=userLineColour](190,65)(190,130)
\psline[linewidth=0.75](215,120)(265,120)
\psline[linewidth=0.75](265,100)(265,120)
\psline[linewidth=0.75](235,100)(265,100)
\psline[linewidth=0.75](225,110)(210,95)
\psline[linewidth=0.75](215,120)(235,100)
\rput(250,110){$\utag$}
\psline[linewidth=0.75](260,100)(260,-5)
\psline[linewidth=0.75](210,95)(210,85)
\rput[tl](220,100){$\pee$}
\psline[linewidth=0.75](240,100)(240,-5)
\rput[l](262.5,0){$\eex$}
\rput[r](237.5,0){$\een$}
\psline[linewidth=0.75](200,75)(200,95)
\psline[linewidth=0.75](200,95)(220,75)
\rput(205,80){$\Thh$}
\psline[linewidth=0.75](200,75)(220,75)
\rput(270,72.5){$\tth$}
\rput(300,100){$\EQLS$}
\psline[linewidth=0.75](380,175)(380,120)
\rput[l](382.5,170){$\eex$}
\newrgbcolor{userLineColour}{1 0.2 0.2}
\psline[linecolor=userLineColour](380,65)(380,90)
\newrgbcolor{userLineColour}{1 0.2 0.2}
\psline[linecolor=userLineColour](320,115)(355,115)
\newrgbcolor{userLineColour}{1 0.2 0.2}
\psline[linecolor=userLineColour](320,65)(380,65)
\newrgbcolor{userLineColour}{1 0.2 0.2}
\psline[linecolor=userLineColour](320,65)(320,115)
\psline[linewidth=0.75](365,120)(395,120)
\psline[linewidth=0.75](395,100)(395,120)
\psline[linewidth=0.75](385,100)(395,100)
\psline[linewidth=0.75](330,105)(350,85)
\rput(387.5,110){$\utag$}
\psline[linewidth=0.75](390,100)(390,-5)
\psline[linewidth=0.75](340,95)(340,85)
\psline[linewidth=0.75](360,85)(360,-5)
\rput[l](392.5,0){$\eex$}
\rput[r](357.5,0){$\een$}
\psline[linewidth=0.75](330,75)(330,95)
\psline[linewidth=0.75](330,95)(350,75)
\rput(335,80){$\Thh$}
\psline[linewidth=0.75](330,75)(350,75)
\psline[linewidth=0.75](352.5,105)(372.5,85)
\psline[linewidth=0.75](365,120)(385,100)
\psline[linewidth=0.75](350,85)(372.5,85)
\psline[linewidth=0.75](330,105)(352.5,105)
\newrgbcolor{userLineColour}{1 0.2 0.2}
\psline[linecolor=userLineColour](380,90)(355,115)
\psline[linewidth=0.75](361.25,96.25)(375,110)
\rput(350,95){$\shh$}
\rput(372.5,72.5){$\uuh$}
\end{pspicture}

%% file: PIC/sHat-WH.tex
\ifx\JPicScale\undefined\def\JPicScale{1}\fi
\psset{unit=\JPicScale mm}
\psset{linewidth=0.3,dotsep=1,hatchwidth=0.3,hatchsep=1.5,shadowsize=1,dimen=middle}
\psset{dotsize=0.7 2.5,dotscale=1 1,fillcolor=black}
\psset{arrowsize=1 2,arrowlength=1,arrowinset=0.25,tbarsize=0.7 5,bracketlength=0.15,rbracketlength=0.15}
\begin{pspicture}(0,0)(415,185)
\psline[linewidth=0.75](55,130)(135,130)
\psline[linewidth=0.75](135,110)(135,130)
\psline[linewidth=0.75](75,110)(135,110)
\psline[linewidth=0.75](65,120)(20,75)
\psline[linewidth=0.75](55,130)(75,110)
\rput(100,120){$\utag$}
\psline[linewidth=0.6](5,130)(35,130)
\psline[linewidth=0.6](35,110)(35,130)
\psline[linewidth=0.6](5,110)(35,110)
\psline[linewidth=0.6](5,130)(5,110)
\psline[linewidth=0.75](95,110)(95,90)
\psline[linewidth=0.6,border=2](20,105)(72.5,82.5)
\psline[linewidth=0.75](20,155)(20,130)
\psline[linewidth=0.6](10,165)(175,165)
\psline[linewidth=0.6](175,145)(175,165)
\psline[linewidth=0.6](10,165)(30,145)
\psline[linewidth=0.75](125,185)(125,165)
\psline[linewidth=0.6](85.62,55)(85.62,75)
\psline[linewidth=0.6](60,75)(60,55)
\psline[linewidth=0.75](20,75)(20,20)
\rput(20,120){$\gtag$}
\rput(100,155){$\iiftag$}
\rput(73.75,65){$\predtag$}
\rput[l](22.5,25){$\pee$}
\rput[t](80,-1.25){$\een$}
\psline[linewidth=0.75](20,110)(20,105)
\psline[linewidth=0.75](80,110)(80,90)
\psline[linewidth=0.6](30,145)(175,145)
\psline[linewidth=0.75](80,40)(80,5)
\psline[linewidth=0.75](100,145)(100,130)
\psline[linewidth=0.6,border=1.5](115,40)(155,55)
\psline[linewidth=0.75,border=0.9](95,5)(95,40)
\newrgbcolor{userFillColour}{0.2 0 0.2}
\rput{90}(72.5,82.5){\psellipse[fillcolor=userFillColour,fillstyle=solid](0,0)(2.58,-2.57)}
\newrgbcolor{userFillColour}{0.2 0 0.2}
\rput{90}(115,40){\psellipse[fillcolor=userFillColour,fillstyle=solid](0,0)(2.58,-2.57)}
\rput[l](102.5,137.5){$\eex$}
\rput[r](77.5,102.5){$\een$}
\rput[r](17.5,137.5){$\pee$}
\rput[l](127.5,180){$\eex$}
\psline[linewidth=0.6](60,75)(85,75)
\psline[linewidth=0.75](155,110)(155,87.5)
\psline[linewidth=0.75](10,10)(30,10)
\psline[linewidth=0.75](10,10)(10,30)
\psline[linewidth=0.75](10,30)(30,10)
\newrgbcolor{userLineColour}{1 0.2 0.2}
\psline[linecolor=userLineColour](180,35)(180,170)
\newrgbcolor{userLineColour}{1 0.2 0.2}
\psline[linecolor=userLineColour](0,170)(180,170)
\newrgbcolor{userLineColour}{1 0.2 0.2}
\psline[linecolor=userLineColour](0,35)(0,170)
\rput(15,15){$\Thh$}
\rput(172.5,42.5){$\thh$}
\rput(195,110){$\EQLS$}
\psline[linewidth=0.75](270,185)(270,130)
\rput[l](272.5,180){$\eex$}
\newrgbcolor{userLineColour}{1 0.2 0.2}
\psline[linecolor=userLineColour](310,60)(310,140)
\newrgbcolor{userLineColour}{1 0.2 0.2}
\psline[linecolor=userLineColour](210,140)(310,140)
\newrgbcolor{userLineColour}{1 0.2 0.2}
\psline[linecolor=userLineColour](210,60)(310,60)
\newrgbcolor{userLineColour}{1 0.2 0.2}
\psline[linecolor=userLineColour](210,60)(210,140)
\psline[linewidth=0.75](235,130)(300,130)
\psline[linewidth=0.75](300,110)(300,130)
\psline[linewidth=0.75](255,110)(300,110)
\psline[linewidth=0.75](245,120)(230,105)
\psline[linewidth=0.75](235,130)(255,110)
\rput(270,120){$\utag$}
\psline[linewidth=0.75](270,110)(270,5)
\psline[linewidth=0.75](230,105)(230,95)
\rput[tl](240,110){$\pee$}
\psline[linewidth=0.75](260,110)(260,5)
\rput[t](260,-1.25){$\een$}
\psline[linewidth=0.75](220,85)(220,105)
\psline[linewidth=0.75](220,105)(240,85)
\rput(225,90){$\Thh$}
\psline[linewidth=0.75](220,85)(240,85)
\rput(302.5,67.5){$\tth$}
\rput(325,110){$\EQLS$}
\psline[linewidth=0.75](400,185)(400,130)
\rput[l](402.5,180){$\eex$}
\newrgbcolor{userLineColour}{1 0.2 0.2}
\psline[linecolor=userLineColour](400,75)(400,100)
\newrgbcolor{userLineColour}{1 0.2 0.2}
\psline[linecolor=userLineColour](340,125)(375,125)
\newrgbcolor{userLineColour}{1 0.2 0.2}
\psline[linecolor=userLineColour](340,75)(400,75)
\newrgbcolor{userLineColour}{1 0.2 0.2}
\psline[linecolor=userLineColour](340,75)(340,125)
\psline[linewidth=0.75](385,130)(415,130)
\psline[linewidth=0.75](415,110)(415,130)
\psline[linewidth=0.75](405,110)(415,110)
\psline[linewidth=0.75](350,115)(370,95)
\rput(407.5,120){$\utag$}
\psline[linewidth=0.75](410,110)(410,5)
\psline[linewidth=0.75](360,105)(360,95)
\psline[linewidth=0.75](380,95)(380,5)
\rput[t](380,-1.25){$\een$}
\psline[linewidth=0.75](350,85)(350,105)
\psline[linewidth=0.75](350,105)(370,85)
\rput(355,90){$\Thh$}
\psline[linewidth=0.75](350,85)(370,85)
\psline[linewidth=0.75](372.5,115)(392.5,95)
\psline[linewidth=0.75](385,130)(405,110)
\psline[linewidth=0.75](370,95)(392.5,95)
\psline[linewidth=0.75](350,115)(372.5,115)
\newrgbcolor{userLineColour}{1 0.2 0.2}
\psline[linecolor=userLineColour](400,100)(375,125)
\psline[linewidth=0.75](381.25,106.25)(395,120)
\rput(370,105){$\shh$}
\rput(392.5,83.75){$\uuh$}
\psline[linewidth=0.75](170,110)(170,87.5)
\psline[linewidth=0.75](150,130)(175,130)
\psline[linewidth=0.75](175,110)(175,130)
\psline[linewidth=0.75](150,110)(175,110)
\psline[linewidth=0.75](150,110)(150,130)
\psline[linewidth=0.75](162.5,145)(162.5,140.62)
\psline[linewidth=0.75](295,110)(295,100)
\psline[linewidth=0.75](285,110)(285,100)
\psline[linewidth=0.75](290,110)(290,85)
\newrgbcolor{userFillColour}{0.2 0 0.2}
\rput{90}(290,95){\psellipse[fillcolor=userFillColour,fillstyle=solid](0,0)(2.58,-2.57)}
\psline[linewidth=0.75](295,100)(290,95)
\psline[linewidth=0.75](285,100)(290,95)
\psline[linewidth=0.6](277.5,85)(302.5,85)
\psline[linewidth=0.6](290,72.5)(302.5,85)
\psline[linewidth=0.6](277.5,85)(290,72.5)
\rput(290,80){$\zzero$}
\rput[t](95,0.62){$\ees$}
\rput[r](92.5,103.75){$\ees$}
\rput[t](270,0.62){$\ees$}
\rput[t](410,1.25){$\ees$}
\psline[linewidth=0.75](130,110)(130,87.5)
\psline[linewidth=0.6](135,55)(135,75)
\psline[linewidth=0.6](90,75)(90,55)
\psline[linewidth=0.75](115,55)(115,25)
\psline[linewidth=0.75](115,110)(115,86.88)
\psline[linewidth=0.75,border=0.9](130,25)(130,55)
\psline[linewidth=0.6](90,75)(135,75)
\psline[linewidth=0.75](122.5,110)(122.5,97.5)
\psline[linewidth=0.75,border=0.8](122.5,55)(122.5,10)
\psline[linewidth=0.6,border=1.5](130,40)(170,55)
\newrgbcolor{userFillColour}{0.2 0 0.2}
\rput{90}(130,40){\psellipse[fillcolor=userFillColour,fillstyle=solid](0,0)(2.58,-2.57)}
\psline[linewidth=0.6](90.62,55)(135,55)
\newrgbcolor{userFillColour}{0.2 0 0.2}
\rput{90}(122.5,17.5){\psellipse[fillcolor=userFillColour,fillstyle=solid](0,0)(2.58,-2.57)}
\psline[linewidth=0.75](130,25)(122.5,17.5)
\psline[linewidth=0.75](115,25)(122.5,17.5)
\psline[linewidth=0.6](110,10)(135,10)
\psline[linewidth=0.6](122.5,-2.5)(135,10)
\psline[linewidth=0.6](110,10)(122.5,-2.5)
\rput(122.5,5){$\zzero$}
\newrgbcolor{userLineColour}{1 0.2 0.2}
\psline[linecolor=userLineColour](0,35)(180,35)
\rput(162.5,120){$\subtract$}
\rput(112.5,65){$\priming$}
\newrgbcolor{userFillColour}{0.2 0 0.2}
\rput{90}(80,40){\psellipse[fillcolor=userFillColour,fillstyle=solid](0,0)(2.58,-2.57)}
\psline[linewidth=0.75](65,55)(80,40)
\psline[linewidth=0.75](110,55)(95,40)
\psline[linewidth=0.75](95,55)(80,40)
\psline[linewidth=0.75,border=1.3](80,55)(95,40)
\newrgbcolor{userFillColour}{0.2 0 0.2}
\rput{90}(95,40){\psellipse[fillcolor=userFillColour,fillstyle=solid](0,0)(2.58,-2.57)}
\psline[linewidth=0.6](60.62,55)(85.62,55)
\psline[linewidth=0.75](95,90)(80,75)
\psline[linewidth=0.75](80,90)(65,75)
\rput[l](131.88,104.38){$\Ints$}
\rput[l](296.88,104.38){$\Ints$}
\rput[l](131.88,27.5){$\Ints$}
\rput(116.88,84.38){$\elprime$}
\rput(131.88,84.38){$\rprime$}
\rput(124.38,93.75){$\mprime$}
\psline[linewidth=0.75](122.5,89.38)(122.5,75)
\rput(155.62,84.38){$\elnoprime$}
\rput(170,84.38){$\rnoprime$}
\psline[linewidth=0.75](155,80)(155,55)
\psline[linewidth=0.75](170,80)(170,55)
\psline[linewidth=0.75](130,80)(130,75)
\psline[linewidth=0.75](115,80)(115,75)
\rput(162.5,137.5){$\rldif$}
\psline[linewidth=0.75](162.5,133.75)(162.5,130.62)
\end{pspicture}

%% file: 7-moncal-outro.tex

A bird's eye view of algebra and coalgebra in computer science suggests that algebra provides \emph{denotational}\/ semantics of computation, whereas coalgebra provides \emph{operational}\/ semantics \cite{KlinB:TCS,PavlovicD:AMAST06,Plotkin-Turi:LICS97}. Denotational semantics goes beyond the purely extensional view of computations (as maps from inputs to outputs), and models certain computational effects (such as non-termination, exceptions, non-determinism, etc.). Operational semantics goes further, and models computational operations. While computational effects are thus presented using the suitable algebraic operations in denotational semantics, computational behaviors are represented as elements of final coalgebras in operational semantics.

But although both the denotational and the operational approaches go beyond the purely \emph{extensional}\/ view, neither has supported a genuinely \emph{intensional}\/ view, envisioned by Turing and von Neumann, where programs are data. Therefore, in spite of the tremendous successes in understanding and systematizing computational structures and behaviors, categorical semantics of computation has remained largely disjoint from theories of computability and complexity.

The claim put forward in this paper is that coalgebra provides a natural categorical framework for a fully intensional categorical theory of computability and complexity. The crucial step that enables this theory leads beyond \emph{final}\/ coalgebras, that assign \emph{unique}\/ descriptions to computational behaviors of \emph{fixed}\/ types, to \emph{universal}\/ coalgebras, that assign \emph{non-unique}\/ descriptions to computations  of \emph{arbitrary}\/ types. These descriptions are what we usually call \emph{programs}. 

Our message is thus that \emph{programmability is a coalgebraic property}, just like \emph{computational behaviors are coalgebraic}. This message is formally expressed through \emph{universal processes}; it can perhaps be expressed more generally through \emph{universal coalgebras}, as families of weakly final coalgebras, all carried by the same \emph{universal state space}. Thm.~\ref{thm:moncom-coalg} spells out in the framework of monoidal computer the fact that every Turing complete programming language provides a universal coalgebra for computable functions of all types; and vice versa, every universal coalgebra induces a corresponding notion of program. Just like abstract computational behaviors of a given type are precisely the elements of a final coalgebra of that type, abstract programs are precisely the elements of a universal coalgebra. Just like final coalgebras can be used to define semantics of computational behaviors \cite{PavlovicD:AMAST06}, universal coalgebras can be used to define semantics of programs. 

From a slightly different angle, the fact that universal coalgebras characterize monoidal computers, proven in Thm.~\ref{thm:moncom-coalg}, can also be viewed as a coalgebraic characterization of computability. There are, of course, many characterizations of computability. The upshot of this one is, however, in Propositions \ref{prop:CX-eval} and \ref{prop:SX-eval}: the coalgebaic view of computability opens an alley towards complexity.
In any universe of computable functions, normal complexity measures \cite{PavlovicD:MonCom2} can be programmed coalgebraically. Combining this coalgebraic view of complexity with the algebraic view of randomized computation seems to open up a path towards a categorical model of one-way functions, and towards categorical cryptography, which has been the original goal of this project \cite{PavlovicD:NSPW11}.

%% file: 8-moncal-appendix.tex
\section*{Appendix}
\begin{proposition*}
Let $\CCC$ be a distributive category, and consider the following structures, characterized by the couniversal properties in Table~\ref{table}:
\begin{enumerate}[a)]
\item exponents $\xxp A B$,
\item final machines $\xxpp A B$.
\end{enumerate}
Then 
\begin{itemize}
\item \textit{(a)} induces \textit{(b)} if $\CCC$ has initial algebras for $F_A(X) = A+(A\times X)$ for all $A$; 
\item \textit{(b)} induces \textit{(a)} if $\CCC$ has absolute limits, i.e. the idempotents split in it.
\end{itemize}
A distributive category with list constructors and split idempotents is thus cartesian closed if and only if it has final machines.
\end{proposition*}

\bpr  Towards the proof of $\mathit{(a)\Longrightarrow (b)}$, suppose that $\CCC$ is a cartesian closed category with the exponents written $\xxp A B$, and with an initial $F_A$-algebra 
\bear
A+\left( A \times A^+\right) & \tto{[\iota, ::]} & A^+
\eear
The intuition is that $A^+$ is the type of nonempty lists of elements of $A$, i.e. the free semigroup generated by $A$. The initial $F_A$-algebra structure consists of the inclusion $A\iinclusion \iota A^+$, and the operation $A\times A^+ \tto{::}$ which can be thought of as prepending a symbol $a\in A$ to the list $\alpha \in A^+$, to construct the list $a :: \alpha$.

The final machine with the inputs from $A$ and the outputs in $B$ is  in the form
\bea\label{eq:finmach}
\xxpp A B \times A & \tto{<\xi_0,\xi_1>} & \xxpp A B \times B
\eea
where $\xi_0$ is derived from prepending $(::)$ and the closed structure by
\[\prooftree
\xxpp A B \times A \times A^+ \tto{\xxpp A B \times (::)} \xxpp A B \times A^+ \tto{\ \varepsilon\ } B
\justifies
\xxpp A B \times A \tto{\xi_0} \xxpp A B
\endprooftree\]
whereas $\xi_1$ is just the evaluation restricted along $\iota$
\[\prooftree
\xxpp A B \times A \tto{\xxp \iota B \times A} \xxp A B \times A \tto{\ \varepsilon\ } B
\justifies
\xxpp A B \times A \tto{\xi_1}  B
\endprooftree\]
To show that \eqref{eq:finmach} is a final machine, note first that every machine $X\times A \tto{<x_0,x_1>} X\times B$ induces an $F_A$-algebra over $\xxp X B$ by transposing
\[\prooftree
X\times \Big(A+ \big( A \times  \xxp X B \big) \Big) \tto{\widetilde \kappa} B
\justifies
A+ \big( A \times  \xxp X B \big) \tto \kappa \xxp X B
\endprooftree\]
where $\widetilde \kappa$ is the composite
\begin{multline*}
X\times \Big(A+ \big( A \times  \xxp X B \big) \Big)
\cong \Big(X\times A\Big) + \Big( X \times A  \times \xxp X B\Big)  \tto{\big(X\times A\big)+\big(x_0 \times \xxp X B \big)} \\
\to  \Big(X\times A\Big) + \Big( X   \times \xxp X B\Big) 
\cong \Big(X\times A\Big) + \Big(\xxp X B \times B\Big)  \tto{[x_1, \varepsilon]} B
\end{multline*}
The $F_A$-algebra $\kappa$ now induces the catamorphism $\cata \kappa$, which induces the anamorphism $\ana x$
$$\begin{tikzar}{}
A+\Big(A\times A^+\Big) \ar{r}{[\iota, ::]}  \ar[dashed]{d}[swap]{A+\left(A\times \cata \kappa\right)} 
\& A^+\ar[dashed]{d}{\cata \kappa}
\&\& X\times A \ar{r}{x} \ar[dashed]{d}[swap]{\ana x \times A} 
\& X\times B \ar[dashed]{d}{\ana x \times B}
\\
A+\Big(A\times \xxp X B \Big) \ar{r}[swap]{\kappa} 
\& \xxp X B
\&\& \xxpp A B \times A \ar{r}[swap]{\xi} 
\& \xxpp A B \times B
\end{tikzar}$$
by the tranposition
\[\prooftree
A^+ \tto{\cata \kappa} \xxp X B
\justifies
X \tto{\ana x} \xxpp A B
\endprooftree\]
The diagram chase showing that the commutativity and uniqueness of the catamorphism on the left induces the commutativity and uniqueness of the anamorphism on the right is lengthy but straightforward.

Towards the proof of $\mathit{(b)\Longrightarrow (a)}$, the assumption is that $\CCC$ has final machines
\bear
\xxpp A B \times A & \tto{<\xi_0, \xi_1>} & \xxpp A B \times B
\eear
so that the machine $\xxpp A B \times A \tto{<\pi_0, \xi_1>}  \xxpp A B \times B$ induces the anamorphism $\ana{\pi_0, \xi_1}$, as displayed on the following diagram.
$$\begin{tikzar}{}
\xxpp A B \times A \ar{rrr}{<\pi_0, \xi_1>} \ar[dashed]{dd}[swap]{\ana{\pi_0,\xi_1} \times A} \ar[two heads]{dr}[swap]{q\times A}
 \&\&\& \xxpp A B \times B \ar[two heads]{dl}{q\times B}
\ar[dashed]{dd}{\ana{\pi_0,\xi_1} \times B}
\\
\& \xxp A B \times A \ar[dashed]{r}{<\pi_0, \mathbf{\varepsilon}>} \ar[tail]{dl}[swap]{m\times A} \& \xxp A B\times B \ar[tail]{dr}{m\times B}\\
\xxpp A B \times A  \ar{rrr}[swap]{<\xi_0, \xi_1>} \&\&\& \xxpp A B \times B
\end{tikzar}$$
Since it is easy to see that $\ana{\pi_0,\xi_1}$ is an endomorphism on the machine  $\xxpp A B \times A \tto{<\pi_0, \xi_1>}  \xxpp A B \times B$, the uniqueness of $\ana{\pi_0,\xi_1}$ as an anamorphism implies 
\bear
\ana{\pi_0, \xi_1} \circ \ana{\pi_0, \xi_1} & = & \ana{\pi_0, \xi_1}
\eear 
Using the assumption that the idempotents in $\CCC$ split, we now define the exponent $\xxp A B$ as the splitting $\xxpp A B\eepi q \xxp A B \mmono m \xxpp A B$ of $\ana{\pi_0, \xi_1}$. The morphism $\xxp A B\times A \tto{\  \varepsilon\ } B$, induced by the splitting in the above diagram, is the counit of the adjunction $(-)\times A \dashv \xxp A -$, for the transposition operation $\lambda$ from Table~\ref{table} defined
\bear
\CCC(X\times A, B) & \tto{\ \lambda\  } & \CCC(X , \xxp A B)\\
f & \longmapsto & \lambda f = q\circ \ana{\pi_0, f}
\eear
To show that $\varepsilon \circ (\lambda f\times A) \ = \ f$ holds, chase the following diagram:
\[\begin{tikzar}{}
X\times A \ar{rrr}{<\pi_0,f>} \ar{dr}{\ana{\pi_0, f}\times A} \ar{dd}[swap]{\lambda f \times A}\& \&\& X\times B \ar{dl}[swap]{\ana{\pi_0, f}\times B} \ar{dd}{\lambda f \times B}\\
\& \xxpp A B \times A \ar{r}{<\pi_0,\xi_1>} \ar[two heads]{dl}{q\times A}\& \xxpp A B \times B \ar[two heads]{dr}[swap]{q\times B}\\
\xxp A B \times A \ar{rrr}[swap]{<\pi_0,\varepsilon>}\&\&\& \xxp A B \times B 
\end{tikzar}
\]
\epr